\documentclass{article}
\usepackage{ashwin}
\usepackage{charter}
\usepackage{tikz}

\usetikzlibrary{arrows.meta,calc,positioning, matrix, shapes.geometric}
\pgfdeclarelayer{edgelayer}
\pgfsetlayers{background,main,edgelayer}
\tikzset{
  pathBox/.style ={fill=red!18, draw=red!70, very thick},
  stepLabel/.style={font=\scriptsize\bfseries,
                    fill=red!80, text=white,
                    circle, inner sep=1.3pt}
}

\providecommand{\email}[1]{\texttt{#1}}

\title{Sparse Navigable Graphs for Nearest Neighbor Search: Algorithms and Hardness} 
\author{Sanjeev Khanna\thanks{University of Pennsylvania \{\email{sanjeev@cis.upenn.edu}, \email{apadaki@seas.upenn.edu}, \email{ewaingar@seas.upenn.edu}\}} \and Ashwin Padaki\footnotemark[1] \and Erik Waingarten\footnotemark[1]}

\begin{document} 

\maketitle 

\begin{abstract}

We initiate the study of approximation algorithms and computational barriers for constructing sparse $\alpha$-navigable graphs~\cite{IX23, DGM24}, a core primitive underlying recent advances in graph-based nearest neighbor search. Given an $n$-point dataset $P$ with an associated metric $\sfd$ and a parameter $\alpha \geq 1$, the goal is to efficiently build the sparsest graph $G=(P, E)$ that is $\alpha$-navigable: for every distinct $s, t \in P$, there exists an edge $(s, u) \in E$ with $\sfd(u, t) < \sfd(s, t)/\alpha$. We consider two natural sparsity objectives: minimizing the maximum out-degree and minimizing the total size (or equivalently, the average degree).

Our starting point is a strong negative result: the slow-preprocessing version of DiskANN (analyzed in~\cite{IX23} for low-doubling metrics) can yield solutions whose sparsity is $\widetilde{\Omega}(n)$ times larger than optimal, even on Euclidean instances. We then show a tight approximation-preserving equivalence between the Sparsest Navigable Graph problem and the classic Set Cover problem, obtaining an $O(n^3)$-time $(\ln n + 1)$-approximation algorithm, as well as establishing NP-hardness of achieving an $o(\ln n)$-approximation. Building on this equivalence, we develop faster $O(\ln n)$-approximation algorithms. The first runs in $\widetilde{O}(n \cdot \mathrm{OPT})$ time and is therefore much faster when the optimal solution is sparse. The second, based on fast matrix multiplication, is a bicriteria algorithm that computes an $O(\ln n)$-approximation to the sparsest $2\alpha$-navigable graph, running in $\widetilde{O}(n^{\omega})$ time.
Finally, we complement our upper bounds with a query complexity lower bound, showing that any $o(n)$-approximation requires examining $\Omega(n^2)$ distances. This result shows that in the regime where $\mathrm{OPT} = \widetilde{O}(n)$, our $\widetilde{O}(n \cdot \mathrm{OPT})$-time algorithm is essentially best possible.

Collectively, these results significantly advance our understanding of the computational complexity of computing sparse navigable graphs.

\end{abstract}

\newpage

\section{Introduction}\label{sec:intro}

This work is on graph-based methods for similarity search, an algorithmic approach which has emerged over the past decade and demonstrated impressive empirical performance~(see~\cite{BIGANN21,BIGANN23}). These techniques aim to solve the \emph{$c$-approximate nearest neighbor search} problem: for an underlying metric space $(X, \sfd)$, the task is to preprocess a dataset $P \subset X$ of $n$ points into a data structure which supports approximate nearest neighbor queries; a new (unseen) query point $q \in X$ is presented, and the data structure must output a dataset point whose distance to $q$ is at most $c$ times that of the closest dataset point. In a graph-based method, the data structure is a directed graph whose vertices correspond to dataset points, and queries are processed by executing a ``greedy-like'' search over the graph.\footnote{The specific implementation varies across heuristics; as in recent work of~\cite{IX23,DGM24}, we analyze the greedy search process (explained formally below).} Until relatively recently, these methods had not been studied from the theoretical computer science perspective, and this work aims to study these methods from the vantage point of \emph{approximation algorithms}. We begin by discussing two recent and related works of~\cite{IX23, DGM24}, which motivate the algorithmic questions at the heart of this paper.

\paragraph{Provable Guarantees for Graph-Based Methods via the Doubling Dimension~\cite{IX23}.}
A feature of graph-based methods is that they are often agnostic to the underlying metric. However, from the worst-case perspective, one may construct metric spaces in which nearest neighbor search requires a prohibitive $\Omega(n)$ query time for any approximation (e.g., a uniform metric~\cite{KL04}). As a result, the community has developed notions of intrinsic dimensionality, which seeks to capture the low-dimensional metric structure in a dataset (see the survey~\cite{C06,DK20}). The notion of intrinsic dimensionality most relevant to nearest neighbor search is the \emph{doubling dimension}~\cite{GKL03, KL04}: a set of points has doubling dimension $\lambda$ if every metric ball can be covered with $2^{\lambda}$ balls of half the radius. For datasets with low doubling dimension, efficient nearest neighbor search is possible. For example, cover trees ~\cite{BKL06} give an $O(n)$-space data structure whose query time is $2^{O(\lambda)} \log n$ for exact nearest neighbor search, and $2^{O(\lambda)} \log \Delta + (1/\eps)^{O(\lambda)}$ for $(1+\eps)$-approximate nearest neighbor search.\footnote{The parameter $\Delta$ is the aspect ratio of the input dataset $P$. It is the maximum pairwise distance between points divided by the minimum nonzero pairwise distance between points.}

Indyk and Xu~\cite{IX23} were the first to parametrize the performance of graph-based methods in terms of the doubling dimension. They study several popular algorithms that are used in practice---DiskANN~\cite{SDSKK19}, NSG~\cite{FXWC19}, and HNSW~\cite{MY18}---and construct datasets with constant doubling dimension (in particular, in $\R^2$) where all of these algorithms require reading $\sim$10\% of the dataset in order to produce an acceptable near neighbor. While seemingly bad news, they also show that DiskANN with \emph{slow-preprocessing time} (removing a heuristic from~\cite{SDSKK19}) performs provably well in metrics with low doubling dimension, with query time $(1/\eps)^{O(\lambda)} \log^2 \Delta$ for $(1+\eps)$-approximate nearest neighbor search. This was the first indication that a theoretical computer science perspective could shed light on the worst-case performance of graph-based methods.\footnote{Prior theoretical works, such as~\cite{L18, PS20}, study graph-based methods from an average-case complexity perspective, when the dataset is generated uniformly at random on a high-dimensional sphere.} Since then, multiple works have sought to prove guarantees and limitations on various aspects of these methods~\cite{DGM24, XSI24, AIKMRSX25, ADGMMS25, GKSW25}. The underlying structural property in~\cite{IX23} as well as the subsequent work~\cite{DGM24} is that of \emph{navigability}, which we discuss next.

\paragraph{Navigability Suffices for Nearest Neighbor Search.} We formalize graph-based methods following~\cite{IX23}. A graph-based method for a dataset $P$, a finite subset of a metric space $(X, \sfd)$, is specified by a directed graph $G$ whose vertices are the elements of $P$. A query $q \in X$ is processed by starting at an arbitrary vertex of $G$ and following a greedy search: at a current vertex $s$, we compute distances $\sfd(q, u)$ for all out-neighbors $u$ of $s$ in the graph $G$; if the minimum $\sfd(q, u)$ is strictly smaller than $\sfd(q, s)$, we update the current vertex to $u$ and repeat; otherwise, we output $s$. The total space of the data structure is the number of edges in $G$. The query time is at most the maximum out-degree of a vertex in $G$, multiplied by the length of the longest ``greedy path'' that may be encountered. The algorithmic challenge, then, is devising a low-degree graph that guarantees fast and accurate convergence of these greedy paths for potentially infinitely-many query points. The key to the analysis in~\cite{IX23} is the following definition, which gives a sufficient condition on the graph $G$ for fast and accurate search.\footnote{In their work,~\cite{IX23} term this notion ``$\alpha$-shortcut reachability.'' We adapt the terminology of~\cite{DGM24} and use the term $\alpha$-navigability.} 
\begin{definition}[$\alpha$-Navigability (Definition~3.1 in~\cite{IX23})]\label{def:navigable} For $\alpha\geq 1$, a directed graph $G = (P, E)$ is said to be \emph{$\alpha$-navigable} if for every $s, t \in P$ with $\sfd(s, t) > 0$, there is some $u \in P$ with $(s, u) \in E$ and $\sfd(u, t) < \sfd(s, t) / \alpha$.
\end{definition}

An important algorithmic contribution is Theorem~3.4 in~\cite{IX23}, which shows that in an $\alpha$-navigable graph $G$, a greedy search always returns an $(\frac{\alpha +1}{\alpha-1} + \eps)$-approximate nearest neighbor in at most $O(\log_{\alpha}(\frac{\Delta}{(\alpha-1) \eps}))$ steps. The final dependence of the query time on the doubling dimension arises from their analysis of the slow-preprocessing algorithm of DiskANN, which constructs an $\alpha$-navigable graph with out-degree at most $O(\alpha)^{\lambda} \log \Delta$. 

The subsequent work of~\cite{DGM24} studies navigability without the assumption of low doubling dimension. Focusing on the case $\alpha = 1$ (the weakest form of $\alpha$-navigability), they show that every $n$-point metric admits a $1$-navigable graph with average degree $O(\sqrt{n \log n})$. Furthermore, they construct subsets of $O(\log n)$-dimensional Euclidean space for which every $1$-navigable graph must have average degree at least $\Omega(n^{1/2 - \delta})$, for any constant $\delta > 0$.

\paragraph{Motivating Questions and Contributions.} Roughly speaking,~\cite{IX23} shows that low doubling dimension is a sufficient condition for the existence of sparse $\alpha$-navigable graphs, while~\cite{DGM24} proves an existential guarantee that every metric admits a nontrivial $1$-navigable graph. Motivated by these results, we propose a broader algorithmic question: given an arbitrary metric and $\alpha\geq 1$, what is the sparsest $\alpha$-navigable graph---and how efficiently can this graph be computed? 

More precisely, we ask:
\begin{itemize}
\item Is there an efficient (approximation) algorithm that, for any metric, constructs the sparsest $\alpha$-navigable graph (minimizing either the maximum out-degree or the total number of edges), thereby yielding instance-by-instance guarantees?
\item Since preprocessing time is a central concern in graph-based nearest neighbor search methods, can we design fast algorithms for building $\alpha$-navigable graphs which do not sacrifice on provable guarantees? 
\end{itemize}
Our work is the first to address these questions from the viewpoint of approximation algorithms. This perspective offers two key advantages: (i) it yields graphs that are (approximately) optimal in sparsity for use in nearest neighbor search, and (ii) when the output graph is dense, it serves as a certificate that no substantially sparser $\alpha$-navigable graphs exist for that instance.

\subsection{Overview of Results}

We provide a comprehensive set of results for computing the sparsest $\alpha$-navigable graph, from both the approximation algorithms and computational hardness perspectives.

\paragraph{Suboptimality of Existing Heuristics (Theorem~\ref{thm:diskann-worst-case}).}
We first show that the slow-preprocessing variant of DiskANN, previously analyzed for low-doubling-dimension metric spaces, can be suboptimal in (nearly) the strongest possible sense. Specifically, we construct explicit subsets of Euclidean space for which this algorithm produces $1$-navigable graphs with $\Omega(n^2)$ edges and maximum out-degree $\Omega(n)$, despite the existence of $1$-navigable graphs with only $O(n \log n)$ edges and maximum out-degree $O(\log n)$. Thus, the DiskANN algorithm incurs a worst-case approximation ratio of $\Omega(n/\log n)$ for both sparsity and degree, even in Euclidean space. Our result bridges a gap in the theoretical analysis of empirical heuristics for graph-based nearest neighbor search.

\paragraph{Equivalence to Set Cover: Approximation Algorithms and Hardness (Theorem~\ref{thm:nav-sc-bounds}).}
We show that the problem of constructing the sparsest $\alpha$-navigable graph is tightly connected to the classic Set Cover problem, by showing an approximation-preserving reduction in both directions. Specifically, for any $\alpha \geq 1$, we prove that:
\begin{itemize}
    \item Any polynomial-time $\rho(n)$-approximation for Set Cover immediately yields a $\rho(n)$-approximation for the sparsest $\alpha$-navigable graph (with respect to both the number of edges and maximum out-degree).
    \item Conversely, a $\rho(n)$-approximation algorithm for the sparsest $1$-navigable graph yields a $\rho(\mathrm{poly}(n))$-approximation for Set Cover.
\end{itemize}
As a consequence, we obtain a polynomial-time $(\ln n + 1)$-approximation algorithm for the sparsest $\alpha$-navigable graph, and show that obtaining a better-than $(c \ln n)$-approximation (for some absolute $c > 0$) is \nphard, even for $\alpha = 1$.

The equivalence above is established by formulating $\alpha$-navigability constraints as covering conditions: for each source vertex, the set of outgoing edges must collectively cover all ``navigation targets'' in accordance with the underlying metric. This enables a direct application of the greedy Set Cover algorithm and its hardness bounds to the navigability setting.

\paragraph{Fast Algorithms for Sparse Instances and Bicriteria Approximation for Dense Instances (Theorems~\ref{thm:fast-log-approx} and~\ref{thm:bicriteria-alg}).}
The preceding connection to the Set Cover problem results in an $O(n^3)$-time algorithm which is not well-suited for large instances. We next show that it is possible to get faster algorithms, both when the solution size is small, and when it is large, albeit with a slight relaxation in requirements:

\begin{itemize}
    \item For instances where the optimal $\alpha$-navigable graph has $\opt$ edges, we give a randomized algorithm that outputs an $O(\ln n)$-approximation in $\Ot(n \cdot \opt)$ time. The key is a {\em membership-query} implementation of the greedy Set Cover algorithm, allowing for an output-sensitive runtime that is sublinear when the minimum cover is small. The membership query model relies on queries of the form {\em ``is element $x$ contained in set $S$''} which, in our setting, corresponds to checking a single $\alpha$-navigability constraint and can be supported in $O(1)$ time.\footnote{To our knowledge, known sublinear-time algorithms for Set Cover~\cite{IMR18,KY14} operate in the adjacency-list model, which assumes sets are presented explicitly as lists of elements. In contrast, our setting is akin to an adjacency-matrix access model, and converting between the two models itself takes linear time.}
    
    \item For general instances, with possibly dense optimal $\alpha$-navigable graphs, we give an algorithm running in $\Ot(n^\omega \log\Delta / \eps)$ time ($\omega$ is the matrix multiplication exponent) that produces an $\alpha$-navigable graph with at most $O(\ln n)$ times the out-degree and number of edges of the sparsest $2\alpha(1+\varepsilon)$-navigable graph, for any $\varepsilon\in(0,1)$. A key insight underlying our algorithm is that we can use fast Boolean matrix multiplication to batch-verify navigability constraints.
\end{itemize}

\paragraph{Lower Bound on Query Complexity (Theorem~\ref{thm:query-lb}).}

We prove a strong lower bound in the black-box metric access model: any algorithm making $o(n^2)$ queries to the metric cannot guarantee even an $o(n)$-approximation to the sparsest $\alpha$-navigable graph problem (under both objectives) even when there is an optimal solution with maximum out-degree $3$. In other words, any non-trivial approximation necessitates examining essentially all pairwise distances in the metric. Thus in the setting of sparse navigable graphs, our bound obtained in Theorem~\ref{thm:fast-log-approx} is essentially best possible.

\paragraph{Very Recent Independent Work.} Very recently, in an independent work, Conway et al.~\cite{CDC25} also studied sparse navigability and its connection to the Set Cover problem. In particular, for the problem of computing an $O(\log n)$-approximation to the sparsest $\alpha$-navigable graph, they give an $\Ot(n^2)$-time algorithm when $\alpha=1$, and an $\Ot(\min\{n^{2.5}, n\cdot\opt\})$-time algorithm when $\alpha > 1$.

\subsection{Organization of the Paper}

Section~\ref{sec:prelims} defines the Sparsest Navigable Graph problem and introduces key preliminaries. Section~\ref{sec:slow-diskann} gives a worst-case approximation lower bound for the slow-preprocessing version of DiskANN. Section~\ref{sec:set-cover} establishes a connection to Set Cover, yielding an approximation algorithm and NP-hardness result. In Section~\ref{sec:upper-bounds}, we present faster approximation algorithms for both sparse and dense instances. Finally, Section~\ref{sec:lower-bounds} proves a strong query complexity lower bound.
\section{Preliminaries}\label{sec:prelims}

\subsection{The $\sparsenav$ Problem}\label{subsec:sparse-nav-graph}

For any directed graph $G = (P, E)$, let $\deg_G(s)$ denote the \emph{out-degree} of vertex $s \in P$.

In the $\sparsenav$ problem, we are given as input a finite metric space $(P,\sfd)$ and a parameter $\alpha\geq 1$. The goal is to output an $\alpha$-navigable graph $G = (P,E)$ that approximates the sparsest $\alpha$-navigable graph on $P$, with respect to either the maximum out-degree or the total number of edges. We formalize our notion of approximation below.

\begin{definition}[Approximate $\sparsenav$]\label{def:approx-nav}
Let $(P,\sfd)$ be a finite metric space and $\alpha\geq 1$. Let $G = (P,E)$ be an $\alpha$-navigable graph.\begin{itemize}
    \item $G$ is a $c$-approximation to the sparsest $\alpha$-navigable graph under the \maxdeg objective if $\max_{s\in P} \deg_G(s)$ is at most $c$ times the maximum out-degree of any $\alpha$-navigable graph on $P$.
    \item $G$ is a $c$-approximation to the sparsest $\alpha$-navigable graph under the \avgdeg objective if $|E|$ is at most $c$ times the number of edges in any $\alpha$-navigable graph on $P$.
\end{itemize}
\end{definition}

Recall that in the worst-case analysis of graph-based nearest neighbor search introduced by~\cite{IX23}, the running time to return an $(\frac{\alpha+1}{\alpha-1}+\eps)$-approximate nearest neighbor query on an $\alpha$-navigable graph $G$ is bounded by \smash{$O(\log_{\alpha}(\frac{\Delta}{(\alpha-1)\eps}))$} times the maximum out-degree of $G$. The space complexity of the corresponding data structure is proportional the average degree of $G$. Therefore, a $c$-approximation to the sparsest $\alpha$-navigable graph with respect to the $\maxdeg$ (resp. $\avgdeg$) objective incurs a $c$-factor overhead in query time (resp. space complexity) relative to the best possible such graph under this analysis.

\subsection{The Set Cover Problem}

\begin{definition}[\setcover]\label{def:set-cover}
    An instance of \setcover is given by a pair $(\calU, \calF)$, where $\calU$ is a set and $\calF$ is a collection of subsets of $\calU$. A cover is a subcollection $\calT\subset \calF$ whose union covers all elements in $\calU$, and the task is to output a cover of minimal size.
\end{definition}

\begin{lemma}[Theorem 4 of~\cite{Joh73}]\label{lem:greedy-set-cover}
    Let $(\calU,\calF)$ be an instance of \setcover with $|\calU| = n$ and $|\calF| = m$. There is an algorithm which returns a $(\ln n + 1)$-approximation to the optimal set cover in time \[O\del{\sum_{S\in\calF} |S|} = O(mn).\]
\end{lemma}
\begin{lemma}[Corollary 4 of~\cite{DS14}]\label{lem:np-hard-set-cover}
    Let $\eps > 0$ be a constant, and let $(\calU,\calF)$ be an instance of \setcover with $|\calU| = n$ and $|\calF| = m = n^{1 + O(1/\eps)} = n^{O(1)}$. It is \nphard to output a $((1-\eps)\cdot\ln n)$-approximation of the size of the minimum set cover.

\end{lemma}
\newcommand{\SlowDiskANN}{\textup{Slow-DiskANN}}

\section{An $\Omt(n)$-Approximation Lower Bound for $\SlowDiskANN$}\label{sec:slow-diskann}

In this section, we examine the slow-preprocessing variant of the DiskANN algorithm~\cite{SDSKK19} studied by Indyk and Xu~\cite{IX23}. This algorithm takes as input a finite metric space $(P,\sfd)$ and constructs an $\alpha$-navigable graph. We show that in the worst case, this particular algorithm outputs an $\Omt(n)$-approximation of the sparsest $\alpha$-navigable graph, under both the $\maxdeg$ and $\avgdeg$ objectives.

\subsection{The $\SlowDiskANN$ algorithm}

$\SlowDiskANN$ builds the out-neighborhood of each $s \in P$ using a greedy covering rule. It repeatedly adds an edge from $s$ to the nearest vertex in $P \setminus \{s\}$ whose $\alpha$-navigability constraint is not yet covered. After an edge is added, it prunes all newly covered vertices, and this process repeats until all vertices are covered.

\begin{algorithm}[H]\label{alg:slow-diskann}
\caption{$\SlowDiskANN$~\cite{IX23}}
\KwIn{Finite metric space $(P,\sfd)$, parameter $\alpha\ge 1$}
\KwOut{$\alpha$-navigable graph $G=(P,E)$}
\ForEach{$s\in P$}{
  Let $U\leftarrow P\setminus\{s\}$ and sort $U$ by increasing $\sfd(s,\cdot)$\;
  \ForEach{$u\in U$ in this order}{
    Update $E\gets E \cup \{(s,u)\}$\;
    Remove from $U$ any $t$ with $\sfd(u,t) < \sfd(s,t)/\alpha$\;
  }
}
\end{algorithm}

\begin{lemma}[Lemma 3.2 of~\cite{IX23}]\label{lem:diskann-navigable}
Algorithm~\ref{alg:slow-diskann} produces an $\alpha$-navigable graph $G = (P,E)$ and runs in $O(n\cdot |E|) = O(n^3)$ time.
\end{lemma}
\begin{proof}
    If $s\neq t$, then either $(s,t)\in E$ or there was some vertex $u\in P\setminus\{s\}$ responsible for pruning the edge $(s,t)$. In the latter case, this implies $(s,u)\in E$ and $d(u,t)\le d(s,t)/\alpha$, thus proving that $G$ is $\alpha$-navigable. The runtime of $G$ is bottlenecked by the pruning procedure, wherein the algorithm performs a linear scan of $U$ every time an edge is added. Hence, the runtime is $O(n\cdot|E|) = O(n^3)$.
\end{proof}

\cite{IX23} also gives a guarantee on the maximum degree of the graph produced by $\SlowDiskANN$, in terms of the \textit{doubling dimension}~\cite{GKL03} of the metric. Specifically, if an $n$-point metric has doubling dimension $\lambda$ and aspect ratio $\Delta$, they show that the $\alpha$-navigable graph returned by $\SlowDiskANN$ has maximum degree at most $O((4\alpha)^{\lambda} \cdot \log \Delta)$.

\subsection{Approximation lower bound for $\SlowDiskANN$}

Despite the above guarantee for $\SlowDiskANN$, we now show that in the worst case, the graph it constructs can be an $\Omt(n)$-approximation to the sparsest $\alpha$-navigable graph, for both the $\maxdeg$ and $\avgdeg$ objectives. Indeed, while~\cite{IX23} proves that metrics with low doubling dimension admit sparse $\alpha$-navigable graphs, we construct high-dimensional Euclidean metrics (which will have large doubling dimension) and admit sparse navigable graphs; yet, $\SlowDiskANN$ on these metrics outputs dense graphs.

\begin{theorem}\label{thm:diskann-worst-case}
$\SlowDiskANN$ (Algorithm~\ref{alg:slow-diskann}) incurs a worst-case approximation factor of $\Omega(n / \log n)$ to the sparsest $\alpha$-navigable graph, under both the $\maxdeg$ and $\avgdeg$ objectives.
\end{theorem}

The proof of Theorem~\ref{thm:diskann-worst-case} relies on a binary tree structure on the unit sphere in $\R^{n}$, defined below. This pointset naturally admits a navigable graph of maximum degree $O(\log n)$, but $\SlowDiskANN$ produces a graph with $\Omega(n^2)$ edges.

\begin{definition}[Euclidean Binary Tree Structure]\label{def:hard-metric-diskann}
    Let $n$ be a power of $2$. We define a set of points $P \subset \mathbb{R}^n$ of size $2n - 1$, consisting of unit vectors under the $\ell_2$-metric, as follows (see also, the accompanying Figure~\ref{fig:nav-graph-bt}):
    \begin{itemize}
        \item For $h\in\{0,1,\ldots,\log_2 n\}$, partition $[n]$ into $n/2^h$ contiguous intervals of size $2^h$. Namely, for $j\in [n / 2^h]$, define \[I_{h,j} := \left((j-1)\cdot 2^h, j\cdot 2^h\right] \cap \Z\] The number of such intervals is $2n-1$.
        \item For each interval $I_{h,j}$, define the corresponding vector $x_{h,j}\in\R^n$ by \[x_{h,j} := \frac{1}{\sqrt{I_{h,j}}}\sum_{i\in I_{h,j}} e_i = \frac{1}{2^{h/2}} \sum_{i\in I_{h,j}} e_i,\] where $e_i$ denotes the $i$-th standard basis vector. Note that by the above normalization, $\norm{x_{h,j}}_2 = 1$.
    \end{itemize}
\end{definition}

Since $P$ consists only of unit vectors, pairwise distances in $P$ can be written in terms of inner products. The next claim characterizes the inner products between vectors in $P$.

\begin{claim}\label{claim:binary-tree-distances}
    For any $x_{h,j}, x_{h',j'}\in P$, we have \[\ip{x_{h,j}}{x_{h',j'}} = \begin{cases}
        0 &\text{if } I_{h,j}\cap I_{h',j'} = \emptyset \\
        2^{-|h-h'| / 2} &\text{otherwise}
    \end{cases}.\]
\end{claim}
\begin{proof}
    Observe that if two intervals $I_{h,j}, I_{h',j'}$ in the above construction are not disjoint, then one is fully contained in the other. If $I_{h,j}\cap I_{h',j'} = \emptyset$, then it is clear that $\ip{x_{h,j}}{x_{h',j'}} = 0$. Otherwise, if $I_{h,j}\subseteq I_{h',j'}$, then $h\leq h'$ and \[\ip{x_{h,j}}{x_{h',j'}} = \frac{2^h}{2^{(h+h')/2}} = 2^{(h-h')/2}.\] The case that $I_{h,j}\supseteq I_{h',j'}$ follows similarly.
\end{proof}

\begin{lemma}\label{lem:exists-sparse-graph-bt}
    There exists a 1-navigable graph $H = (P,E)$ with maximum out-degree $O(\log n)$.
\end{lemma}
\begin{proof}
    We construct a navigable graph $H=(P,E)$ as follows: \begin{itemize}
        \item Draw the (directed) binary tree rooted at $x_{\log_2 n, 1}$: for all $h\in\{1,2,\ldots,\log_2 n\}$ and $j\in[n/2^h]$, we include the edges $(x_{h,j}, x_{h-1, 2j-1})$ and $(x_{h,j}, x_{h-1, 2j})$.
        \item Draw an edge from every node to each of its ancestors in the above tree: for all $h,h'\in\{0\}\cup[\log_2 n]$ and indices $j\in[n/2^h]$ and $j'\in[n/2^{h'}]$ where $I_{h,j}\subset I_{h',j'}$, draw the edge $(x_{h,j}, x_{h',j'})$.
    \end{itemize}

    The graph $H$ is shown in Figure~\ref{fig:nav-graph-bt}. Since each $I_{h,j}$ is contained in $(\log_2 n) - h$ intervals $I_{h',j'}$, it follows that $H$ has maximum out-degree $\log_2 n + 1 = O(\log n)$.

    We now check navigability of $H$. For $s = x_{h, j}$ and $t = x_{h', j'}$, let $x_{h^\ast, j^\ast}$ denote the least common ancestor in the binary tree, i.e. $I_{h^\ast,j^\ast}$ is the smallest interval containing both $I_{h,j}$ and $I_{h',j'}$. \begin{itemize}
        \item If $x_{h',j'} = x_{h^\ast, j^\ast}$ then $I_{h, j}\subset I_{h',j'}$, which means $(x_{h,j},x_{h',j'})\in E$.
        \item If $x_{h,j} = x_{h^\ast, j^\ast}$, then for a child $u\in \cbr{x_{h^\ast -1, 2j^\ast - 1}, x_{h^\ast - 1, 2j^\ast}}$ of $x_{h^\ast,j^\ast}$, we have $(x_{h,j},u)\in E$ and $\norm{u-x_{h',j'}}_2 < \norm{x_{h,j}-x_{h',j'}}_2$ by Claim~\ref{claim:binary-tree-distances}.
        \item Otherwise, it must be that $I_{h, j}\cap I_{h', j'} = \emptyset$. For $u = x_{h^\ast, j^\ast}$, then, we have $(s,u)\in E$ and $\norm{u-t}_2 < \sqrt{2} = \norm{s-t}_2$.
    \end{itemize}
    Thus, it follows that $H$ is navigable.
\end{proof}

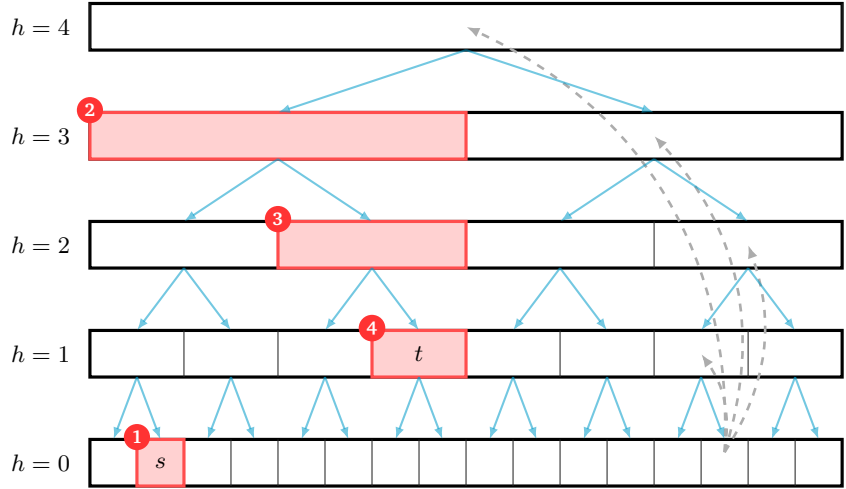
\begin{figure}
    \begin{center}
        
    \begin{tikzpicture}[
  >=Latex,
  levelRect/.style={draw,very thick,fill=white},
  segmentLine/.style={draw=black!85},
  pcEdge/.style={cyan!80!black,opacity=0.55,thick,-{Latex[length=1.6mm]}},
  ancEdge/.style={gray!65,dashed,line width=1pt,-{Latex[length=1.8mm]}},
  stEdge/.style={red!80!black,opacity=0.4,very thick,-{Latex[length=2.2mm]}},
  lbl/.style={font=\small},
  labelStyle/.style={font=\small\itshape}
]

\def\W{10}      
\def\H{0.62}    
\def\VS{1.45}   

\def\SegList{0/16,1/8,2/4,3/2,4/1}

\foreach \lv/\segs in \SegList {
  \pgfmathsetmacro\cw{\W/\segs}
  \coordinate (BL\lv) at (0,\lv*\VS);    
  \coordinate (TR\lv) at (\W,\lv*\VS+\H);
  \draw[levelRect] (BL\lv) rectangle (TR\lv);

  \pgfmathtruncatemacro\Last{\segs-1}
  \foreach \j in {1,...,\Last}{
      \draw[segmentLine] (\j*\cw,\lv*\VS) -- (\j*\cw,\lv*\VS+\H);
  }

  \node[labelStyle,left=4pt] at (0,\lv*\VS+\H/2) {$h=\lv$};

  \foreach \j in {0,...,\Last}{
      \coordinate (C\lv-\j) at ({(\j+0.5)*\cw}, {\lv*\VS+\H/2});
  }
}

\begin{scope}
  \foreach \lv/\j/\step in {
      0/1/1,   
      3/0/2,   
      2/1/3,   
      1/3/4    
  }{
    \pgfmathsetmacro{\cw}{\W/(2^(4-\lv))}

    \path[pathBox]
      ($(C\lv-\j)+(-0.5*\cw,-0.5*\H)$)
      rectangle
      ($(C\lv-\j)+( 0.5*\cw, 0.5*\H)$);

    \node[stepLabel]
      at ($(C\lv-\j)+(-0.5*\cw,0.55*\H)$) {\step};
  }
\end{scope}

\def\slevel{0}\def\sidx{1}    
\def\tlevel{1}\def\tidx{3}    
\def\rlevel{0}\def\ridx{13}   

\node[lbl] at (C\slevel-\sidx) {$s$};
\node[lbl] at (C\tlevel-\tidx) {$t$};

\foreach \lv in {1,2,3,4}{
  \pgfmathtruncatemacro\pow{int(2^\lv)}
  \pgfmathtruncatemacro\anc{int(\ridx/\pow)}
  \pgfmathsetmacro\bend{18 + 4*\lv}   
  \draw[ancEdge] ([yshift=4]C\rlevel-\ridx) to[bend right=\bend] (C\lv-\anc);
}

\foreach \lv/\segs in {4/1,3/2,2/4,1/8}{
  \pgfmathtruncatemacro\clvl{\lv-1}       
  \pgfmathtruncatemacro\Last{\segs-1}     
  \foreach \j in {0,...,\Last}{
    \pgfmathtruncatemacro\cA{2*\j}        
    \pgfmathtruncatemacro\cB{2*\j+1}      

    \draw[pcEdge] 
      ($(C\lv-\j) + (0,-0.5*\H)$) -- ($(C\clvl-\cA) + (0,0.5*\H)$);

    \draw[pcEdge] 
      ($(C\lv-\j) + (0,-0.5*\H)$) -- ($(C\clvl-\cB) + (0,0.5*\H)$);
  }
}

\end{tikzpicture}
    \end{center}

    \caption{
    Depiction of a sparse $1$-navigable graph on the pointset of Definition~\ref{def:hard-metric-diskann} with maximum degree $O(\log n)$, for $n=16$.
    The horizontal bands represent ``levels'' of the binary tree $h = 0$ to $h = 4$.
    Blue edges denote parent–child edges in the tree;
    grey dashed edges point from a node to its ancestors;
    the shaded red cells and their circled numbers depict the traversal of a navigable path from $s = x_{0,2}$ to $t = x_{1,4}$.
    }
    \label{fig:nav-graph-bt}
\end{figure}

\begin{lemma}\label{lem:diskann-dense-graph-bt}
    Running Algorithm~\ref{alg:slow-diskann} on $P$ with $\alpha=1$ produces a graph $G$ containing $\Omega(n^2)$ edges.
\end{lemma}
\begin{proof}
    We will show that for all $i\in[n]$, we have $\deg_G(x_{0,i})$, encoding the size of the out-degree of $x_{0, i}$ in $G$ being $\Omega(n)$. Consider the execution of Algorithm~\ref{alg:slow-diskann} with $s = x_{0,1}$, corresponding to the interval $I_{0,1}=\{1\}$. Since all vectors in $P$ have unit norm, sorting by increasing distance to $x_{0,1}$ is equivalent to storing by decreasing inner product $\ip{x_{0,1}}{\cdot}$.  By Claim~\ref{claim:binary-tree-distances}, the first $\log_2 n$ entries of the sorted list $U$ will be ordered \[ x_{1,1}, x_{2,1}, \ldots, x_{\log_2 n, 1}.\] The remaining entries in $U$ are equidistant to $x_{0,1}$. Algorithm~\ref{alg:slow-diskann} is deterministic, and orders vertices solely based on distance; hence, consider the next $n-1$ entries being \[x_{0,2},\ldots,x_{0,n}.\] The algorithm will begin by adding the edge $(x_{0,1}, x_{1,1})$, which will immediately prune from $U$ all points $x_{h,1}$ for $h > 1$, since \[\ip{x_{0,1}}{x_{h,1}} = 2^{-h/2} < 2^{-(h-1)/2} = \ip{x_{1,1}}{x_{h,1}}.\] 

    Moreover, as $I_{1,1} = \{1,2\}$, the algorithm will prune $x_{0,2}$ from $U$, but will not prune any $x_{0,j}$ for $j>2$. After the first round of pruning, then, the next $n-2$ entries in $U$ will be  \[\cbr{x_{0,3},\ldots,x_{0,n}}.\] Since all distances in $\{x_{0,1},\ldots,x_{0,n}\}$ are equal, the algorithm will add all edges of the form $(x_{0,1}, x_{0,j})$ for $j>2$. Hence, $\deg_G(x_{0,1}) = n-2$, and by symmetry of $P$, it follows that $\deg_G(x_{0,i})=n-2$ for all $i\in[n]$. Hence, $G$ has $\Omega(n^2)$ edges.    
\end{proof}

We are now ready to prove Theorem~\ref{thm:diskann-worst-case}.

\begin{proof}[Proof of Theorem~\ref{thm:diskann-worst-case}]
    Consider the set $P\subset\R^n$ from Definition~\ref{def:hard-metric-diskann} under the Euclidean metric. By Lemma~\ref{lem:exists-sparse-graph-bt}, there exists a 1-navigable graph on $P$ with maximum out-degree $O(\log n)$ and $O(n\log n)$ edges. However, by Lemma~\ref{lem:diskann-dense-graph-bt}, the graph output by Algorithm~\ref{alg:slow-diskann} has $\Omega(n^2)$ edges and maximum out-degree $\Omega(n)$. It follows that Algorithm~\ref{alg:slow-diskann} returns an $\Omega(n/\log n)$-approximation to the sparsest navigable graph on $P$, for both the $\maxdeg$ and $\avgdeg$ objectives.
\end{proof}

\begin{remark}
    The $\Omega(n^2)$ lower bound in Lemma~\ref{lem:diskann-dense-graph-bt} relies on the assumption that the sorting of $U$ breaks ties in an adversarial manner. This behavior can be enforced deterministically by slightly perturbing the construction: for instance, by shrinking the bottom-level vectors $x_{0,j}$ by a factor of $(1 - \varepsilon)$ for sufficiently small $\varepsilon > 0$. It can be checked that this modification preserves Lemma~\ref{lem:exists-sparse-graph-bt}, while ensuring that Algorithm~\ref{alg:slow-diskann} produces a graph with $\Omega(n^2)$ edges regardless of its tiebreaking rule.
\end{remark}

\newcommand{\PUF}{P_{(\calU,\calF)}}

\section{Sparse Navigability is Equivalent to Set Cover}\label{sec:set-cover}

In this section, we show that the problem of constructing the sparsest $\alpha$-navigable graph on a given metric is equivalent, in terms of its polynomial-time {\em approximability threshold}, to the classic \setcover problem (Definition~\ref{def:set-cover}). This connection yields a polynomial-time $(\ln n + 1)$-approximation algorithm for constructing the sparsest $\alpha$-navigable graph and a matching hardness of approximation result (up to a constant factor).

\begin{theorem}\label{thm:nav-sc-bounds} Let $(P,\sfd)$ be an $n$-point metric, and let $\alpha\geq 1$.
\begin{itemize}
    \item There is a polynomial-time algorithm that produces a $(\ln n + 1)$-approximation to the sparsest $\alpha$-navigable graph on $P$ (under both the $\maxdeg$ and $\avgdeg$ objectives).
    \item There is a constant $c\in(0,1)$ such that it is $\nphard$ to produce a $(c\ln n)$-approximation to the sparsest $\alpha$-navigable graph on $P$ (under either the $\maxdeg$ or $\avgdeg$ objectives), even when $\alpha = 1$.
\end{itemize}
\end{theorem}

The proof of Theorem~\ref{thm:nav-sc-bounds} relies on a two-way approximation-preserving reduction between $\sparsenav$ and $\setcover$. Lemma~\ref{lem:nav-sc-red} shows that any polynomial-time $\rho(n)$-approximation algorithm for $\setcover$ on a universe of size $n$ can be used to construct a $\rho(n)$-approximation to the sparsest $\alpha$-navigable graph on an $n$-point metric. Conversely, Lemma~\ref{lem:sc-nav-red} maps a $\setcover$ instance with $n$ elements and $m$ sets to a metric of size $N = (m+n)^{\Theta(1)}$, and shows that a $\rho(N)$-approximation to the sparsest $1$-navigable graph yields a $\Theta(\rho(N))$-approximation to the original $\setcover$ instance. Since the polynomial-time approximation threshold for $\setcover$ is $\Omega(\ln N) = \Omega(\ln (m+n))$, it follows that $\sparsenav$ inherits the same approximation threshold up to a constant factor.

\subsection{Reduction from Sparsest Navigable Graph to Set Cover}

We begin by expressing $\alpha$-navigability as a set covering condition.

\begin{definition}[Set Cover Formulation of $\alpha$-Navigability]\label{def:sc-instance-nav}
    Let $(P,\sfd)$ be a metric space and $\alpha \geq 1$. For each $s, u\in P$ with $u\neq s$, let $Z_{\alpha}(s,u)$ denote the set of $\alpha$-navigability constraints covered by the edge $(s,u)$. Formally, we define \[Z_{\alpha}(s,u) := \cbr{t\in P \mid \sfd(u,t) < \sfd(s,t) / \alpha} \subseteq P \setminus\{s\}.\]
    For each $s\in P$, we define an instance $(P\setminus\{s\}, \calF_s)$ of $\setcover$, where $\calF_s := \cbr{Z_{\alpha}(s,u) \mid u\in P \setminus\{s\}}$. We remark that the collection $\calF_s$ can be built in time $O(n^2)$.    
\end{definition}

\begin{claim}\label{claim:nav-sc-red}
    The following statements are equivalent: 
    \begin{itemize}
        \item $G = (P,E)$ is $\alpha$-navigable.
        \item For every $s\in P$, $\cbr{Z_{\alpha}(s,u) \mid u\in N^{\textup{out}}(s)}$ forms a valid set cover for the instance $(P\setminus\{s\}, \calF_s)$.
    \end{itemize}
\end{claim}
\begin{proof}
    If $G$ is $\alpha$-navigable, then for all $s\in P$ and $t\in P\setminus\{s\}$, there is some $u\in N^{\textup{out}}(s)$ for which $t\in Z_{\alpha}(s,u)$. If $G$ is not $\alpha$-navigable, then there is some $s\in P$ and $t\in P\setminus\{s\}$ such that $t\notin Z_{\alpha}(s,u)$ for all $u\in N^{\textup{out}}(s)$.
\end{proof}

The above equivalence allows us to reduce the task of producing an $\alpha$-navigable graph on $P$ to solving $|P|$ instances of $\setcover$. As a result, any approximation algorithm for $\setcover$ immediately yields an approximation algorithm for constructing sparse $\alpha$-navigable graphs, with the same approximation factor.

\begin{lemma}\label{lem:nav-sc-red}
Suppose there exists a polynomial-time $\rho(N)$-approximation algorithm for $\setcover$ on instances $(\calU, \calF)$ with $|\calU| \leq N$. Then, for any $\alpha \geq 1$, there exists a polynomial-time algorithm that computes a $\rho(n)$-approximation to the sparsest $\alpha$-navigable graph on any $n$-point metric space, under both the $\maxdeg$ and $\avgdeg$ objectives.
\end{lemma}

\begin{proof}
    Let $\calA$ be a $\rho(N)$-approximation algorithm for solving $\setcover$ on universes of size at most $N$. Given an $n$-point metric space $P\subset(X,\sfd)$ and a parameter $\alpha\geq 1$, we describe the following algorithm for building an $\alpha$-navigable graph $G=(P,E)$. We begin with an initially empty graph $G$, and for each $s\in P$, we run the following process: \begin{enumerate}
        \item Compute the set system $\calF_s := \cbr{Z_{\alpha}(s,u) \mid u \in P\setminus\{s\}}$.
        \item Run $\calA$ on the $\setcover$ instance $(P\setminus\{s\}, \calF_s)$ to obtain a solution $\cbr{Z_{\alpha}(s,u_1), \ldots, Z_{\alpha}(s, u_{\hat{k}_s})}$.
        \item For $1\leq i\leq \hat{k}_s$, add the edge $(s, u_i)$ to $G$.
    \end{enumerate}
Since each $\calF_s$ can be built in time $O(n^2)$ and $\calA$ runs in polynomial time, the above algorithm runs in polynomial time. By Claim~\ref{claim:nav-sc-red}, the graph obtained by the above algorithm $G$ is $\alpha$-navigable. To finish, we now argue that $G$ is a $\rho(n)$-approximation to the sparsest $\alpha$-navigable graph on $P$, under both the $\maxdeg$ and $\avgdeg$ objectives. 

Take any $s \in P$ and let $H$ be an $\alpha$-navigable graph on $P$. By Claim~\ref{claim:nav-sc-red}, the out-neighborhood of $s$ in $H$ yields a solution to the $\setcover$ instance $(P \setminus \{s\}, \mathcal{F}_s)$ of size $\deg_H(s)$ (here, $\deg_{H}(s)$ and $\deg_G(s)$ denote the number of out-going edges of $s$ in $H$ and $G$, respectively). Since $|P\setminus\{s\}| < n$, the approximation guarantee of $\mathcal{A}$ implies that it returns a solution of size $\hat{k}_s\leq \rho(n)\cdot \deg_H(s)$, which yields \[\deg_G(s)\leq \rho(n)\cdot\deg_H(s).\] Since this bound holds for all $s\in P$, we have $\max_s\deg_G(s) \leq \rho(n)\cdot\max_{s}\deg_H(s)$, and $|E(G)|\leq \rho(n)\cdot |E(H)|$. Therefore, $G$ is a $\rho(n)$-approximation to the sparsest $\alpha$-navigable graph on $P$ under both the $\maxdeg$ and $\avgdeg$ objectives.
\end{proof}

\subsection{Reduction from Set Cover to Sparsest Navigable Graph}

Suppose we are given an instance $(\calU,\calF)$ of $\setcover$, where $\calU = \{x_1,\ldots,x_n\}$ and $\calF = \{S_1,\ldots,S_m\}$. We show how to construct, in polynomial time, a metric space $(\PUF,\sfd)$ such that approximations to the sparsest $\alpha$-navigable graph on $\PUF$ (under either the $\maxdeg$ or $\avgdeg$ objective) correspond to approximations to the optimal set cover on $(\calU,\calF)$.

\begin{definition}[Metric Construction for $(\calU,\calF)$]\label{def:sc-metric}
    We define a weighted graph and take $(\PUF,\sfd)$ to be the induced shortest-path metric. The construction will depend on two parameters: a small constant $\gamma > 0$, and an integer $L\geq m+n$. We build the graph as follows (depicted in Figure~\ref{fig:sc-metric}):

    \begin{itemize}
        \item Create $L$ distinct root vertices $r_1,\ldots,r_L$ and $L$ identical gadgets encoding the $\setcover$ instance $(\calU,\calF)$.
        \item For each $q\in[L]$, the $q$-th gadget $P^{(q)}$ will consist of $m+n$ vertices corresponding to each $S_i\in \calF$ and $x_j\in\calU$, labeled $S_i^{(q)}$ and $x_j^{(q)}$ respectively.
        \begin{itemize}
            \item[$\circ$] Every root vertex $r_\ell$ will have an edge $(r_\ell, S_i^{(q)})$ of weight $1$. Additionally, for each $i\neq i'$, we draw an edge $(S_i^{(q)}, S_{i'}^{(q)})$ of weight $1-\gamma$.
            \item[$\circ$] Each inclusion $x_j\in S_i$ will correspond to an edge $(S_i^{(q)}, x_j^{(q)})$ of weight $1$. Also, every root vertex $r_\ell$ will have a direct ``shortcut'' edge $(r_\ell, x_j^{(q)})$ of weight $2-\gamma$.
        \end{itemize}
    \end{itemize}
    
    The resulting metric space $(\PUF,\sfd)$ has size $|\PUF| = L + L\cdot |P^{(q)}|= L\cdot (m+n+1)$.
\end{definition}

\begin{figure}
    \centering
    \begin{tikzpicture}[>=Latex,
      root/.style   ={circle,draw,thick,fill=white,inner sep=2pt,font=\small},
      setV/.style   ={rounded corners = 1mm,draw=cyan!80!black,thick,
                      fill=white,minimum width=12mm,minimum height=5mm,font=\small},
      elem/.style   ={circle,draw=gray!65,thick,fill=white,inner sep=2pt,font=\small},
      rootEdge/.style   ={cyan!80!black,thick,-},
      membership/.style ={cyan!80!black,thick,-},
      setEdge/.style    ={gray!80!black,thick,-},
      nonedge/.style    ={gray!55,dashed,thin,opacity=.4},
      shortcut/.style   ={red!80!black,very thick,-},
      lbl/.style        ={font=\scriptsize}
    ]
    \node[root] (r) at (0,4) {$r$};
    \foreach \i/\x in {1/-4, 2/0, 3/4}{
      \node[setV] (S\i) at (\x,2) {$S_{\i}$};
      \draw[rootEdge] (r) -- (S\i);
    }
    \draw[setEdge, text opacity = 1] (S1) -- node[pos=.70,lbl,yshift=5pt] {$1-\gamma$} (S2);
    \draw[setEdge] (S2) -- (S3);
    
    \foreach \j/\x in {1/-6, 2/-2, 3/0, 4/2, 5/6}{
      \node[elem] (x\j) at (\x,0) {$x_{\j}$};
    }
    
    \draw[membership, text opacity = 1] (S1) -- node[pos=.55,lbl,above left=-2pt] {$1$} (x1);
    \draw[membership] (S1) -- (x2);
    \draw[membership] (S2) -- (x2);
    \draw[membership] (S2) -- (x3);
    \draw[membership] (S2) -- (x4);
    \draw[membership] (S3) -- (x4);
    \draw[membership] (S3) -- (x5);
    
    \draw[shortcut, opacity=0.55, text opacity = 1] (r) to[out=135,in=120,looseness=1.1] node[pos=.55,lbl,above left=-3pt] {$2-\gamma$} (x1);
    \draw[shortcut, opacity=0.55] (r) to[out=150,in=100,looseness=1.05] (x2);
    \draw[shortcut, opacity=0.55] (r) to[out=30,in=80,looseness=1.05]  (x4);
    \draw[shortcut, opacity=0.55] (r) to[out=45,in=60,looseness=1.1]   (x5);

    \end{tikzpicture}
    \caption{
    Weighted subgraph induced by a single gadget and one root
    vertex $r$ in the metric space of
    Definition~\ref{def:sc-metric} for a $\setcover$ instance with $S_1 = \{x_1,x_2\}$, $S_2 = \{x_2,x_3,x_4\}$, and $S_3 = \{x_4,x_5\}$.
    Blue edges (weight~1) connect $r$ to each set vertex and each set
    vertex to the universe elements it contains; red edges (weight $2-\gamma$) connect $r$ to every element vertex; and grey edges (weight $1-\gamma$) connect set vertices. For visual clarity, the edges $(S_1,S_3)$ and $(r,x_3)$ are omitted.}
    \label{fig:sc-metric}
\end{figure}
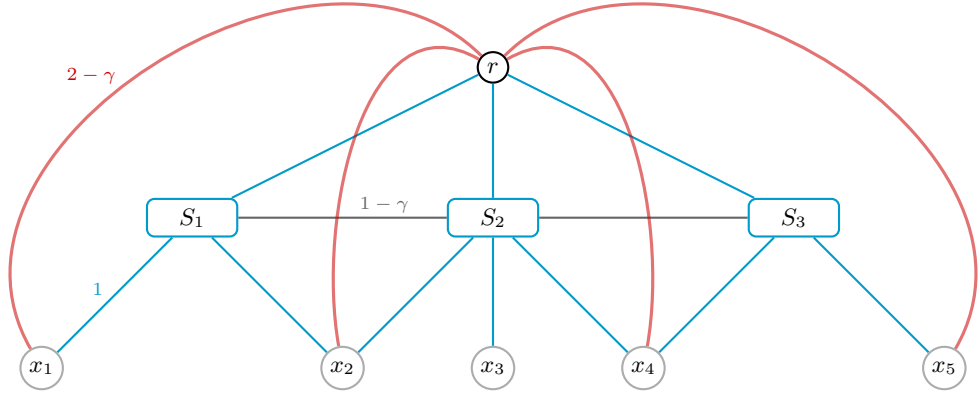

To see how the $\setcover$ instance is encoded in the above construction, consider a single gadget $P^{(q)}$ and the navigability constraint from a fixed root vertex $r_\ell$ to a vertex \smash{$x_j^{(q)}$}. Observe that \smash{$\sfd(S_i^{(q)}, x_j^{(q)}) < \sfd(r_\ell, x_j^{(q)}) = 2-\gamma$} if and only if $x_j \in S_i$. Using this property, one can show that satisfying navigability from the root vertex $r_\ell$ is roughly equivalent to constructing a valid set cover of $\calU$. By duplicating the root vertices and gadgets by a factor of $L$, we magnify the contribution of the optimal cover size of $(\calU,\calF)$ towards both the $\maxdeg$ and $\avgdeg$ of the sparsest navigable graph on $\PUF$. The formal argument appears below.

\begin{claim}\label{claim:small-nav-from-sc}
    If $(\calU,\calF)$ has a solution of size $\opt$, then the metric $(\PUF,\sfd)$ admits a $1$-navigable graph of maximum out-degree  $\max\{L\cdot\opt, m+n\}$ and with $L^2\cdot\opt+L(m+n)^2$ edges.
\end{claim}

\begin{proof}
    Let $\{S_{i_1},\ldots,S_{i_\opt}\}$ be a solution to $(\calU,\calF)$. We define the a $1$-navigable graph $G = (\PUF,E)$. We will describe the subgraph induced by the root vertices $r_1,\ldots,r_L$ and a single gadget $P^{(q)}$, and the entire graph will be the union of these identical subgraphs over all gadgets. \begin{itemize}
        \item For $\ell\in[L]$ and each $b\in[\opt]$, add the edge $(r_\ell, S_{i_b}^{(q)})$.
        \item Include the complete graph on $P^{(q)}$, and add an edge from every vertex in $P^{(q)}$ to $r_1$ (this will ensure we satisfy the navigability constraints among vertices in different gadgets).
    \end{itemize}

    The out-degree of each root vertex $r_\ell$ is precisely $L\cdot \opt$, and the out-degree of each $S_i^{(q)}$ and $x_j^{(q)}$ is $m+n$. Therefore, $G$ has max out-degree $\max\{L\cdot\opt, m+n\}$ and $|E| = L^2\cdot \opt + L(m+n)^2$, as desired. It remains to argue that $G$ is 1-navigable. Let $s,t\in \PUF$, and divide into a small number of cases. \begin{itemize}
        \item If $s,t$ are contained in the same gadget $P^{(q)}$, then the direct edge $(s,t)\in E$ satisfies the constraint.
        \item If $s\in P^{(q)}$ and $t\notin P^{(q)}$ is not $r_1$, then $(s,r_1)\in E$ and $\sfd(r_1, t) \leq 2-\gamma < 2 \leq \sfd(s,t)$.
        \item If $s=r_\ell$ and $t=r_{\ell'}$, then $(s,S_{i_1}^{(1)})\in E$ and $\sfd(S_{i_1}^{(1)}, r_{\ell'}) = 1 < 2 = \sfd(r_\ell, r_{\ell'})$.
        \item If $s = r_\ell$ and $t = S_i^{(q)}$, then since $\opt\geq 1$, we have $(r_\ell, S_{i'}^{(q)})\in E$ for some $i'\in[m]$, and \[\sfd(S_{i'}^{(q)}, S_i^{(q)}) = 1-\gamma < 1 = \sfd(r_\ell, S_i^{(q)}).\]
        \item Finally, if $s=r_\ell$ and $t = x_j^{(q)}$, then there is some $b\in[\opt]$ where $x_j\in S_{i_b}$. Thus, $(r_\ell, S_{i_b}^{(q)})\in E$ and \[\sfd(S_{i_b}^{(q)}, x_j^{(q)}) = 1 <2-\gamma = \sfd(r_\ell, x_j^{(q)}).\]
    \end{itemize}
    Hence, $G$ is $1$-navigable.
\end{proof}

\begin{claim}\label{claim:small-sc-from-nav}
    Let $\opt$ be the size of the minimum set cover for $(\calU,\calF)$. Then, in any $1$-navigable graph $G$ on $(\PUF,\sfd)$, all root vertices $r_1,\ldots,r_L$ have out-degree at least $L\cdot \opt$.
\end{claim}
\begin{proof}
    Let $G = (\PUF,E)$ be $1$-navigable. We show that for every root vertex $r_\ell$ and every gadget $P^{(q)}$, the out-neighborhood of $r_{\ell}$, denoted $N^{\textup{out}}(r_{\ell})$ satisfies $|N^{\textup{out}}(r_{\ell})\cap P^{(q)}|\geq \opt$, which implies the claim. Fix $r_\ell$ and $P^{(q)}$, and consider any element vertex $x_j^{(q)}$. Since $G$ is $1$-navigable, the navigability constraint from $r_\ell$ to $x_j^{(q)}$ means there is an edge $(r_\ell, p)\in E$ such that $\sfd(p,x_j^{(q)}) < \sfd(r_\ell, x_j^{(q)}) =2-\gamma$. The metric is defined so that only two conditions may arise, and we build an assignment $\sigma \colon N^{\textup{out}}(r_{\ell}) \cap P^{(q)} \to \calF$ which gives a set cover for $(\calU, \calF)$ accordingly.
    \begin{itemize}
        \item Either $p = x_j^{(q)}$, in which case we let $\sigma(p) = S_i \in \calF$.
        \item Or, $p = S_i^{(q)}$ for some set $S_i\in\calF$ containing $x_j$. In this case, we let $\sigma(p) \in \calF$ be an arbitrary set containing $x_j$ (such a set must exist or $(\calU, \calF)$ does not admit any cover).
    \end{itemize}     
    The case that $p$ is another node can be ruled out. If \smash{$p' = x_{j'}^{(q)}$} for $j' \neq j$ or \smash{$p' = x_{j'}^{(q')}$} for $q' \neq q$, then \smash{$\sfd(p', x_{j}^{(q)}) \geq 2$}. Note, the resulting assignment from the image of $\sigma \subset \calF$ forms a cover of $\calU$ of size at most $|N^{\textup{out}}(r_{\ell}) \cap P^{(q)}|$, which implies $|N^{\textup{out}}(r_{\ell}) \cap P^{(q)}| \geq \opt$. Since $P^{(q)}$ and $P^{(q')}$ are disjoint for $q \neq q'$ and there are $L$ such indices, the out-degree of each root $r_{\ell}$ is of size at least $L \cdot \opt$.
\end{proof}

\begin{lemma}\label{lem:sc-nav-red}
    Suppose there exists a polynomial-time algorithm that computes a $\rho(N)$-approximation to the sparsest $1$-navigable graph on metric spaces of size at most $N$, under either the $\maxdeg$ or $\avgdeg$ objectives. Then, for any instance $(\calU, \calF)$ of $\setcover$ with $|\calU| = n$ and $|\calF| = m$, there exists a polynomial-time algorithm that computes a $2\rho((n + m + 1)^3)$-approximation to the size of the minimum set cover.
\end{lemma}

\begin{proof}
    Let $\calA$ be a polynomial-time $\rho(N)$-algorithm for computing the sparsest $1$-navigable graph (under the $\maxdeg$ or $\avgdeg$ objective) on metrics of size at most $N$. We describe an algorithm for approximating the size $\opt$ of the optimal solution of an instance $(\calU,\calF)$ of $\setcover$. \begin{enumerate}
        \item Build the metric $(\PUF,\sfd)$ from Definition~\ref{def:sc-metric} with $L = (m+n+1)^2$.
        \item Run $\calA$ on $(\PUF,\sfd)$ to obtain a $1$-navigable graph $G = (\PUF,E)$.
        \begin{enumerate}
            \item If $\calA$ approximates the $\maxdeg$ objective, choose an arbitrary root vertex $r_\ell$ and return $\deg_G(r_\ell)$, the out-degree of $r_{\ell}$ in $G$, divided by $L$.
            \item If $\calA$ approximates the $\avgdeg$ objective, return $|E| / L$.
        \end{enumerate}
    \end{enumerate}
    
    Since $(\PUF,\sfd)$ can be constructed in polynomial time and $\calA$ runs in polynomial time, this algorithm runs in polynomial time. We claim that the above algorithm returns a $2\rho((m+n+1)^3)$-approximation of $\opt$. First, suppose that $\calA$ approximates with respect to the $\maxdeg$ objective, and observe that by Claim~\ref{claim:small-nav-from-sc} and $L\geq m+n$, there is a $1$-navigable graph on $\PUF$ of max out-degree at most $L\cdot\opt$. Therefore, \[\deg_G(r_\ell)\leq \rho(|\PUF|)\cdot L\cdot\opt = \rho((m+n+1)^3)\cdot L\cdot\opt\]
    Also, by Claim~\ref{claim:small-sc-from-nav},  $1$-navigability of $G$ implies that $\deg_G(r_\ell) \geq L\cdot \opt$, and thus \[\frac{\deg_G(r_\ell)}{L} \in \left[\opt, \rho((m+n+1)^3)\cdot\opt\right].\]

    Next, suppose that $\calA$ approximates with respect to the $\avgdeg$ objective, and observe that by Claim~\ref{claim:small-nav-from-sc} and the setting of $L = (m+n+1)^2$, there is a navigable graph on $\PUF$ with at most $2\cdot(m+n+1)^4 \cdot\opt= 2L^2\cdot\opt$ edges. Therefore, \[|E| \leq \rho((m+n+1)^3)\cdot 2L^2\cdot \opt.\] By Claim~\ref{claim:small-sc-from-nav}, we must have $|E|\geq L^2\cdot\opt$, and thus \[\frac{|E|}{L^2}\in\left[\opt, 2 \rho((m+n+1)^3)\cdot\opt\right],\] as desired.
\end{proof}

\begin{remark}
    Although Lemma~\ref{lem:sc-nav-red} only establishes an approximation to the size of the optimal set cover, the reduction can be extended to recover an approximately optimal set cover itself using a standard averaging argument. 
\end{remark}

We are now ready to prove the main theorem of this section.

\begin{proof}[Proof of Theorem~\ref{thm:nav-sc-bounds}]
    We begin with the first statement. By Lemma~\ref{lem:greedy-set-cover}, there is a polynomial-time $(\ln N + 1)$-approximation algorithm for $\setcover$ on instances $(\calU,\calF)$ with $|\calU| = N$. Combining with Lemma~\ref{lem:nav-sc-red}, this implies that for any $n$-point metric $(\PUF,\sfd)$ and $\alpha\geq 1$, there is a polynomial-time algorithm from computing a $(\ln n + 1)$-approximation to the sparsest $\alpha$-navigable graph on $\PUF$, under both the $\maxdeg$ and $\avgdeg$ objectives.
    
    We now show that it is $\nphard$ to compute a $(c \ln N)$-approximation to the sparsest $\alpha$-navigable graph (under either the max out-degree or average-degree objectives) on metrics of size $N$, for $\alpha = 1$ and a constant $c \in (0,1)$ to be specified later. Suppose that there exists a polynomial-time algorithm achieving such an approximation. Then by Lemma~\ref{lem:sc-nav-red}, this would imply a polynomial-time algorithm for approximating the optimal set cover on instances $(\calU,\calF)$ within factor
    \[
    2c \cdot \ln(n + m + 1)^3 \le 6c \cdot \ln(n + m + 1),
    \]
    where $|\mathcal{U}| = n$ and $|\mathcal{F}| = m$. However, by Lemma~\ref{lem:np-hard-set-cover} (with $\varepsilon = 1/2$), it is $\nphard$ to approximate the minimum set cover to within factor $\frac{1}{2} \ln n$, even when $|\mathcal{U}| = n$ and $|\mathcal{F}| \le n^D$ for some constant $D>1$. Since $n + m + 1 \le 2n^D$, we have:
    \[
    \ln(n + m + 1) \le \ln(2n^D) = D \ln n + \ln 2 \le 2D \ln n,
    \]
    for sufficiently large $n$. Thus, the approximation factor $6c \ln(n + m + 1)$ is at most $12cD \ln n$. Thus, for the setting of $c = \frac{1}{24D}$, it is $\nphard$ to compute a $(c\ln N)$-approximation to the sparsest navigable graph, as desired.
\end{proof}

\begin{remark}\label{rem:n^3-bound}
    The $(\ln n + 1)$-approximation algorithm from the first part of Theorem~\ref{thm:nav-sc-bounds} runs in time $O(n^3)$. This is because the reduction in Lemma~\ref{lem:nav-sc-red} takes $O(n^2)$ time to construct each of the $n$ instances of $\setcover$, and the greedy algorithm from Lemma~\ref{lem:greedy-set-cover} solves each of these instances in time $O(n^2)$. The $O(n^3)$ runtime matches that of Slow-DiskANN (Lemma~\ref{lem:diskann-navigable}), but with a vastly better worst-case approximation guarantee: while Slow-DiskANN can incur an $\Omega(n / \log n)$-approximation in the worst case (Theorem~\ref{thm:diskann-worst-case}), the set-cover-based algorithm guarantees a $(\ln n + 1)$-approximation for all inputs.

\end{remark}
\newcommand{\OPT}{\mathrm{OPT}}
\newcommand{\contains}{\mathrm{contains}}
\newcommand{\Fail}{\mathtt{FAIL}}
\newcommand{\FindHeavySet}{\mathrm{FindHeavySet}}
\newcommand{\FastSetCover}{\mathrm{FastSetCover}}
\newcommand{\Ualive}{\calU^{\textup{alive}}}
\newcommand{\Unew}{\calU^{\textup{new}}}

\newcommand{\BuildNav}{\mathrm{BuildNav}}
\newcommand{\VerifyNav}{\mathrm{VerifyNav}}

\SetKw{Continue}{continue}

\section{Fast Approximation Algorithms for Sparse Navigability}\label{sec:upper-bounds}

Theorem~\ref{thm:nav-sc-bounds} gives a $(\ln n + 1)$-approximation algorithm for the sparsest $\alpha$-navigable graph, running in $O(n^3)$ time, and it shows that improving this approximation to $o(\ln n)$ is $\nphard$. A natural question is whether the $O(n^3)$ runtime can be improved for a similar approximation guarantee. In the first main result of this section, we show that if the sparsest $\alpha$-navigable graph on a given metric has $\OPT_{\textup{size}}$ edges, then we can compute an $O(\ln n)$-approximation in $\wt{O}(n\cdot\OPT_{\textup{size}})$ time. This yields up to a near-linear speedup for input metrics which admit sparse $\alpha$-navigable graphs.

\begin{theorem}\label{thm:fast-log-approx}
Let $(P,\sfd)$ be an $n$-point metric space and $\alpha\geq 1$. \begin{itemize}
    \item There exists a randomized algorithm which outputs an $O(\ln n)$-approximation to the sparsest $\alpha$-navigable graph on $P$ (under both the $\maxdeg$ and $\avgdeg$ objectives), with high probability.
    \item Letting $\OPT_{\textup{size}}$ denote the minimum number of edges in any $\alpha$-navigable graph on $P$, the algorithm runs in time $\Ot(n \cdot \OPT_{\textup{size}})$.
\end{itemize}
\end{theorem}

The key ingredient in the proof of Theorem~\ref{thm:fast-log-approx} is Lemma~\ref{lem:fast-greedy-cover} (stated below), which provides a $O(\ln n)$-approximation algorithm for $\setcover$ in the membership-query access model. The benefit (as opposed to directly utilizing Lemma~\ref{lem:greedy-set-cover}) is that the algorithm of Lemma~\ref{lem:fast-greedy-cover} finds a cover faster if the optimal cover is small.

In the second main result of this section, we show that potentially faster algorithms are possible, with running time $\tilde{O}(n^{\omega} \log \Delta / \eps)$ for any $\eps \in (0,1)$, even when $\OPT_{\textup{size}}$ is large. The speedup is achieved by allowing a bicriteria approximation (i.e., incorporating a relaxation of the navigability parameter $\alpha$). Specifically, we give an $\wt{O}(n^\omega \log \Delta / \eps)$-time algorithm that produces an $\alpha$-navigable graph whose maximum out-degree and number of edges are within an $O(\ln n)$-factor of those of the sparsest $2\alpha(1+\eps)$-navigable graph.

\begin{definition}[Bicriteria Approximation to the Sparsest Navigable Graph]\label{def:bi-criteria}
Let $(P,\sfd)$ be a finite metric space and $\beta\geq\alpha\geq 1$. We say that $G = (P,E)$ is a $(\alpha,\beta)$-\emph{bicriteria} $c$-approximation to the sparsest $\alpha$-navigable graph under the $\{\maxdeg, \avgdeg\}$ objective if $G$ is $\alpha$-navigable and:
\begin{itemize}
    \item ($\maxdeg$) $\max_{s\in P} \deg_G(s)$ is at most $c$ times the max out-degree of any $\beta$-navigable graph on $P$.
    \item ($\avgdeg$) $|E|$ is at most $c$ times the number of edges in any $\beta$-navigable graph on $P$.
    \end{itemize}
\end{definition}

\begin{remark}
    For $\beta \geq \alpha \geq 1$,\, $\beta$-navigability imposes tighter constraints than $\alpha$-navigability. Namely, any edge $(s, u)$ satisfying the $\beta$-navigability constraint from $s$ to $t$, i.e., $\sfd(u, t) < \sfd(s, u) / \beta$, also satisfies the corresponding $\alpha$-navigability constraint. Hence, the sparsest $\beta$-navigable graph will contain at least as many edges as the sparsest $\alpha$-navigable graph, meaning that Definition~\ref{def:bi-criteria} is indeed a relaxation of the Sparsest Navigable Graph problem.
\end{remark}

\begin{theorem}\label{thm:bicriteria-alg}
Let $(P, \sfd)$ be an $n$-point metric space with aspect ratio $\Delta$, and let $\alpha \geq 1$.

\begin{itemize}
    \item For any $\eps \in (0,1)$, there exists a randomized algorithm which outputs an $(\alpha,\, 2\alpha(1+\eps))$-bicriteria $O(\ln n)$-approximation to the sparsest navigable graph on $P$ (under both the $\maxdeg$ and $\avgdeg$ objectives), with high probability.
    \item This algorithm runs in time $\Ot(n^\omega \log \Delta / \eps)$.
\end{itemize}
\end{theorem}

\subsection{Faster Sparse Navigability via Membership Set Cover}

Definition~\ref{def:sc-instance-nav} and Lemma~\ref{lem:nav-sc-red} reduce the construction of a sparse $\alpha$-navigable graph on an $n$-point metric to solving $n$ instances of $\setcover$, where each potential edge $(s,u)$ corresponds to a set $Z_\alpha(s,u)$ of navigability constraints it would satisfy. In the standard formulation of $\setcover$ (Definition~\ref{def:set-cover}), each set $S \in \calF$ is represented explicitly as a list of elements from the universe $\calU$. Under this representation, the runtime of the reduction is $O(n^3)$, since there are $\binom{n}{2}$ such sets to construct, each taking up to $O(n)$ time to build. However, each set $Z_{\alpha}(s, u)$ is implicitly defined by the distances in $\sfd$, so we design algorithms for \setcover which implicitly work with $Z_{\alpha}(s, u)$ via membership queries (defined below).

\begin{definition}[\setcover in the Membership-Query Model]
    An algorithm for \setcover in the \textit{membership-query} model accesses the input $(\calU,\calF)$ only through (constant-time) queries of the form \[\contains(S, x) := \begin{cases}
        1 & \text{if } x\in S \\
        0 & \text{if } x\notin S
    \end{cases}\] for $x\in \calU$ and $S\in\calF$.
\end{definition}

For $n=|\calU|$ and $m=|\calF|$, the following lemma gives an algorithm computing an $O(\ln n)$-approximation to the minimum set cover in the membership-query model that runs in time $\Ot(mk + nk)$, where $k$ is the size of the minimum cover. Note that for $k \ll n$, this runtime is much faster than $O(mn)$, the number of membership queries needed to explicitly construct all sets in $\calF$.

\begin{lemma}\label{lem:fast-greedy-cover}
There is a randomized algorithm in the membership-query model (Algorithm~\ref{alg:fast-set-cover}) that, given an instance $(\calU,\calF)$ of \setcover with $|\calU| = n$, $|\calF| = m$, and minimum cover size $k$, runs in time $\Ot(mk + nk)$ and computes an $O(\ln n)$-approximate set cover with probability at least $1-O(1/n^2)$.
\end{lemma}

The proof of the lemma is deferred to Subsection~\ref{subsec:sc-lem}. We now prove Theorem~\ref{thm:fast-log-approx}, assuming Lemma~\ref{lem:fast-greedy-cover}.

\begin{proof}[Proof of Theorem~\ref{thm:fast-log-approx}]
    Let $(P,\sfd)$ be an $n$-point metric and $\alpha\geq 1$, and let $\calA$ be the membership-query algorithm for $\setcover$ from Lemma~\ref{lem:fast-greedy-cover}. For each $s \in P$, define the set system $\calF_s := \{ Z_{\alpha}(s,u) \mid u \in P \setminus \{s\} \}$, as in Definition~\ref{def:sc-instance-nav}. By Claim~\ref{claim:nav-sc-red}, a graph is $\alpha$-navigable if and only if for each $s \in P$, the out-neighbors $N^{\text{out}}(s)$ form a set cover for $(P\setminus\{s\}, \calF_s)$.

    The algorithm for building an $\alpha$-navigable graph proceeds as follows: run $\calA$ on each instance $(P\setminus\{s\}, \calF_s)$ using only membership queries, and construct the graph corresponding to the returned solutions to these instances. For each $u,t \in P\setminus\{s\}$, the function $\contains(Z_{\alpha}(s,u), t)$ tests whether $\sfd(u,t) < \sfd(s,t)/\alpha$, which takes constant time. Each instance succeeds with probability $1-O(1/n^2)$, so all $n$ instances succeed with probability $1-o(1)$, which we will assume for the remainder of the proof. We now prove the two approximability guarantees.

    For each $s\in P$, let $k_s$ denote the minimum out-degree of $s$ in any $\alpha$-navigable graph on $P$. It is clear that $\sum_s k_s \leq \OPT_{\textup{size}}$. By Claim~\ref{claim:nav-sc-red}, $k_s$ is also the size of the minimum set cover for the instance $(P\setminus\{s\},\calF_s)$. By Lemma~\ref{lem:fast-greedy-cover}, the algorithm $\calA$ runs in $\Ot(n k_s)$ time for each $s$, returning an $O(\ln n)$-approximate set cover of size $O(k_s \ln n)$. Thus, in time \[\sum_{s\in P} \Ot(nk_s) = \Ot(n\cdot \OPT_{\textup{size}}),\] we produce an $\alpha$-navigable graph with maximum degree \[\max_{s\in P} O(k_s\ln n) = O(\ln n)\cdot \max_{s\in P} k_s\] and size \[\sum_{s\in P} O(k_s\ln n) = O(\ln n)\cdot \OPT_{\textup{size}}.\] Therefore, the output of the algorithm is an $O(\ln n)$-approximation to the sparsest $\alpha$-navigable graph on $P$, under both the $\maxdeg$ and $\avgdeg$ objectives.
\end{proof}

\subsubsection{Proof of Lemma~\ref{lem:fast-greedy-cover}}\label{subsec:sc-lem}

\paragraph{Algorithm Overview.} Let $(\calU,\calF)$ be an instance of $\setcover$ with $n=|\calU|$ and $m=|\calF|$, and let $k$ denote the size of the minimum set cover. The classic greedy algorithm (Lemma~\ref{lem:greedy-set-cover}) proceeds in rounds, and selects at each round a set in $\calF$ that covers at least a $1/k$-fraction of the uncovered elements (such a set is guaranteed to exist because there is always cover of size at most $k$). After $r$ rounds, there are at most $(1-1/k)^r n$ uncovered elements, guaranteeing a cover once $r > k\ln n$. Our aim is to simulate the above greedy approach via random sampling as efficiently as possible. 

We design a subroutine $\FindHeavySet$ (Algorithm~\ref{alg:find-heavy-set}) that simulates a round of the greedy set cover algorithm. The subroutine takes as input a parameter $\hat{k}$, which serves as a guess for the size of the optimal cover $k$, and uses it to determine the number of randomly drawn sets and elements needed. Specifically, it samples random collections of sets from $\calF$ and estimates their sizes by performing membership queries against a random sample of elements from $\Ualive$, the set of currently uncovered elements. The sampling parameters are chosen so that $\FindHeavySet$ makes only $\Ot(m)$ membership queries and, if $\hat{k} \geq k$, then with high probability it returns a set covering an $\Omega(1/k)$-fraction of the remaining uncovered elements, which is sufficient for the desired $O(\ln n)$-approximation.

The main algorithm, $\FastSetCover$ (Algorithm~\ref{alg:fast-set-cover}), maintains the set $\Ualive$ of uncovered elements, which is initialized to $\calU$. In each iteration, it attempts to find a set covering a sufficiently large fraction of $\Ualive$ by calling $\FindHeavySet$ with a current guess $\hat{k}$ for the optimal cover size $k$. If a set $\bS\in\calF$ is found, it updates $\Ualive$ to $\Ualive\setminus\bS$. If no set is found, $\hat{k}$ is doubled and the process repeats. Once $\Ualive$ is empty, the algorithm returns the selected sets.

\begin{algorithm}[H]\label{alg:fast-set-cover}
    \caption{$\FastSetCover(\calU,\calF)$}
    \KwIn{a \setcover instance $(\calU,\calF)$ in the membership-query model}
    \KwOut{a set cover $\calC$ of $(\calU,\calF)$}
    \BlankLine
    Initialize $\calC = \{\}$ and $\Ualive = \calU$\;
    Initialize $\hat{k} = 1$\;
    \While{$\calU\neq\emptyset$}{
        $\bS\gets\FindHeavySet(\Ualive, \calF, \hat{k}, |\calU|)$\;
        \If{$\bS = \Fail$}{
            Update $\hat{k} \gets 2\cdot \hat{k}$\;
            \Continue
        }
        Update $\calC \gets \calC\cup\{\bS\}$\;
        Update $\Ualive \gets \Ualive\setminus \bS$\;
    }
    \Return $\calC$\;
        
\end{algorithm}

\begin{algorithm}[H]\label{alg:find-heavy-set}
    \caption{$\FindHeavySet(\Ualive,\calF, \hat{k}, n)$}
    \KwIn{set of uncovered elements $\Ualive\subseteq\calU$, collection of subsets $\calF\subset 2^{\calU}$, parameter $\hat{k}$ }
    \KwOut{a set $\bS\in\calF$ with large intersection $|\bS\cap \Ualive|$}
    \BlankLine

    $m = |\calF|$\;
    \For{$i=1$ to $\ceil{\log_2 \hat{k}}$}{
        Sample $R = \ceil{m\log n \cdot 2^{-(i-2)}}$ sets $\bS_1,\ldots,\bS_R\simiid \Unif(\calF)$\;
        Sample $T = 48\log (mn) \cdot \min\cbr{2^{i+3}\log n, 2\hat{k}}$ elements $\bx_1,\ldots,\bx_T\simiid\Unif(\Ualive)$\;
        \For{$r=1$ to $R$}{
            \If{$\#\{t\mid \contains(\bS_r,\bx_t) = 1\}\geq 24\log(mn)$}{
                \Return $\bS_r$\;
            }
        }
    }

    \Return $\Fail$\;
\end{algorithm}

We start by analyzing the performance of $\FindHeavySet$. We begin with an elementary fact about sequences with bounded sum, which will help us reason about the distribution of set sizes in an optimal cover.

\begin{restatable}{claim}{harmonicbalancing}
\label{claim:harmonic-balancing}
    Let $\alpha_1\geq \alpha_2\geq\cdots\geq\alpha_L\geq 0$ be non-negative numbers with $\sum_{j=1}^{L}\alpha_j = A$.
    Then, there exists an index $\ell\in [L]$ such that $\ell\cdot \alpha_{\ell} \geq A / H_L$, where $H_L$ denotes the $L$-th harmonic number.

\end{restatable}
\begin{proof}
    Deferred to Appendix~\ref{appendix:missing-proofs-ub}.
\end{proof}

Next, we quantify the number of uniform samples needed to estimate the size of a given set.

\begin{restatable}{claim}{sizetesting}
\label{claim:size-testing}
Let $S \subset [N]$ be fixed, and let $\bX := \{x_1, \dots, x_T\}$ where $x_i \simiid \Unif([N])$. Then for all $\alpha \in [0,1]$, the following bounds hold with probability at least $1 - \exp(-\alpha T/12)$:
\begin{enumerate}[label=(\roman*)]
    \item If $|S| \geq \alpha N$, then $|S \cap \bX| \geq \frac{\alpha T}{2}$.
    \item If $|S| \leq \frac{\alpha N}{4}$, then $|S \cap \bX| < \frac{\alpha T}{2}$.
\end{enumerate}
\end{restatable}
\begin{proof}
    Deferred to Appendix~\ref{appendix:missing-proofs-ub}.
\end{proof}

Using Claims~\ref{claim:harmonic-balancing} and~\ref{claim:size-testing}, we prove a desired guarantee on the performance of $\FindHeavySet$.

\begin{lemma}\label{lem:heavy-set-properties}
    Let $(\calU,\calF)$ be an instance of $\setcover$ with $|\calU| = n$, $|\calF| = m$, and with minimum cover size $k$. For $\Ualive\subseteq\calU$, The algorithm $\FindHeavySet(\Ualive,\calF, \hat{k}, n)$ runs in time $\Ot(m)$\footnote{Omitting $\log$-factors in both $m$ and $n$} and, with probability at least $1-O(1/n^3)$, has the following properties.
    \begin{itemize}
        \item If $\hat{k}\geq k$, then $\FindHeavySet$ does not return $\Fail$
        \item If $\FindHeavySet$ does not return $\Fail$, then it returns a set $\bS\in\calF$ satisfying \[|\bS\cap \Ualive| \geq |\Ualive|/(8\hat{k}).\] 
    \end{itemize}
\end{lemma}
\begin{proof}
    We first bound the runtime of $\FindHeavySet(\Ualive, \calF, \hat{k}, n)$. The outer loop of $\FindHeavySet$ runs over at most $\log_2 k \leq O(\log n)$ iterations. In each iteration, we make $R+T=O(RT)$ random samples, each of which take $O(\log n)$ time by Claim~\ref{claim:fast-uncovered-sampling}; and we make $RT$ queries to $\contains(\cdot,\cdot)$. Since $RT = O(m\log^2 n \log(mn))$, the total runtime is $O(RT\log^2 n) = O(m\log^3 n \log(mn)) = \Ot(m)$.

    We now show the desired properties. For each set $S\in\calF$, denote $\wt{S} := S\cap \Ualive$, and write $N = |\Ualive|$. By assumption, there is a set cover for $(\Ualive, \calF)$ of size $k$, meaning there are $k$ sets $S_1,\ldots,S_k\in \calF$ with $|\wt{S}_1| \geq \dots \geq |\wt{S}_k|$ such that \[|\wt{S}_1| + \dots + |\wt{S}_k| \geq N.\] 
    
    Let $L$ be the largest index in $[k]$ for which $|\wt{S}_L| \geq N/(2k)$. Then, we have \[|\wt{S}_1| + \dots + |\wt{S}_L|\geq N/2.\] By Claim~\ref{claim:harmonic-balancing}, there is some $\ell\in [L]$ for which $\ell\cdot |\wt{S}_\ell| \geq N / (2H_L) > N / (4\log n)$, where we have used that $2H_L\leq 2(\log L + 1) < 4\log n$. This implies the following key property: 
    \begin{align}
        \text{For some $\ell\leq k$, there are $\ell$ sets $S\in\calF$ with $|\wt{S}| \geq N\cdot \max\cbr{\frac{1}{4\ell\log n}, \frac{1}{2k}}$} \tag{$\ast$} .\label{eq:heavy-level}
    \end{align}
    
    We now prove the first bullet point. When $k\leq \hat{k}$, there is an iteration $i\leq \ceil{\log_2 \hat{k}}$ of the outer loop of $\FindHeavySet$ for which $\ell\in [2^i, 2^{i+1}]$. We argue that with high probability, if the algorithm reaches iteration $i$, it will terminate during this iteration and avoid returning $\Fail$. By~(\ref{eq:heavy-level}), there are at least $2^i$ sets $S\in\calF$ with $|\wt{S}| \geq \max\cbr{\frac{N}{2^{i+3}\log n}, \frac{N}{2k}}$. By the choice of $R \geq \frac{3m\log n}{2^i}$, the probability that at least one of these sets $S$ are sampled in $\cbr{\bS_1,\ldots,\bS_R}$ is \[1-(1-2^i / m)^{R} \leq 1-\exp(-2^i R / m) \geq 1-1/n^3.\] We condition on this event that $\bS_r = S$ for some $r\in[R]$. Suppose we have not terminated when the inner loop is at iteration $r$, and let $\bX = \cbr{\bx_1,\ldots,\bx_T}\subset \Ualive$ denote the random sample during this iteration. We observe that $|\wt{S}|\geq \alpha N$ for \[\alpha = \max\cbr{\frac{1}{2^{i+3}\log n}, \frac{1}{2k}} = \frac{48\log (mn)}{T}.\] By the first statement of Claim~\ref{claim:size-testing} we conclude that, with probability $1-\exp(-\frac{\alpha T}{12}) = 1-1/(mn)^4$, we have $|S\cap\bX| = |\wt{S}\cap\bX|\geq 18\log n$, and we will return $S$. In total, then, $\FindHeavySet$ avoids returning $\Fail$ with probability at least $1-1/(mn)^4-1/n^3\geq 1-2/n^3$.
    
    We will now show the second bullet point. Suppose that $\FindHeavySet$ does not return $\Fail$, and instead the outer loop of $\FindHeavySet$ terminates at some iteration $i\leq\log_2 \hat{k}$. Recall that we only return a set $S\in\calF$ if $|S\cap \bX|=|\wt{S}\cap \bX|\geq 18\log n = \frac{\alpha T}{2}$, which holds for the setting of $\alpha = \frac{36\log(mn)}{T} \geq \frac{1}{2\hat{k}}$. If $|\wt{S}| < \alpha N / 4$, then by the second statement of Claim~\ref{claim:size-testing} implies that $|\wt{S}\cap\bX| < \frac{\alpha T}{2}$ with probability $1-\exp(-\alpha T /12) = 1-1/(mn)^4$. Union bounding over the $R\log_2 \hat{k} \ll mn$ sets sampled throughout the execution of the algorithm, we conclude that, with probability $1-1/(mn)^3$, the set $\bS$ returned by $\FindHeavySet$ satisfies \[|\wt{\bS}| \geq \frac{\alpha N}{4}  \geq \frac{N}{8\hat{k}}\]
    
    Overall, we conclude that the properties in both bullet points hold with probability $1-2/n^3 - 1/(mn)^3 = 1 - O(1/n^3)$, as desired.
\end{proof}

Finally, we show how in the execution of $\FastSetCover$ we may efficiently maintain the set of uncovered elements $\Ualive$.

\begin{restatable}{claim}{fastuncoveredsampling}\label{claim:fast-uncovered-sampling}
Let $\calU$ be a set of size $n$. There is a data structure that can be built in time $O(n)$ which maintains a subset $\Ualive \subseteq \calU$ under deletions, and supports:
\begin{itemize}
    \item checking whether $\Ualive=\emptyset$ in $O(1)$ time,
    \item deletion of any element from $\Ualive$ in $O(\log n)$ time, and
    \item uniform random sampling from $\Ualive$ in $O(\log n)$ time
\end{itemize}
\end{restatable}
\begin{proof}
    Deferred to Appendix~\ref{appendix:missing-proofs-ub}.
\end{proof}

We are now ready to prove the main guarantee for $\FastSetCover$ stated in Lemma~\ref{lem:fast-greedy-cover}.

\begin{proof}[Proof of Lemma~\ref{lem:fast-greedy-cover}]
    We condition on the event that every call to $\FindHeavySet$ satisfies the two properties of Lemma~\ref{lem:heavy-set-properties}. The first property implies that in the main loop of $\FastSetCover$, at most $1+\log_2 k$ calls to $\FindHeavySet$ will return $\Fail$, and $\hat{k}\leq 2k$ throughout the execution of the algorithm. The second property implies that whenever $\FindHeavySet$ does not return $\Fail$, it will return a set $S$ with \[|S\cap\Ualive| \geq \frac{|\Ualive|}{8\hat{k}} \geq \frac{|\Ualive|}{16k}.\] Therefore, at every iteration of the main loop, we have \[|\Ualive|\leq n \cdot \del{1-\frac{1}{16k}}^{|\calC|} \leq n\cdot \exp(-|\calC| / 16k)\] and thus upon termination, $|\calC|\leq 1 + 16k\ln n = O(k\ln n)$, proving the desired $O(\ln n)$-approximation guarantee. 
    
    To bound the runtime, we observe that $\FastSetCover$ makes at most $1 + \log_2 k + O(k\log n) = O(k\log n)$ calls to $\FindHeavySet(\Ualive, \calF, \hat{k}, n)$, each with parameter $\hat{k}\leq 2k$. By the runtime guarantee of Lemma~\ref{lem:heavy-set-properties}, the time needed to execute these calls is $\Ot(mk)$. The only other nontrivial computation in $\FastSetCover$ is the maintenance of the data structure for $\Ualive$. By Claim~\ref{claim:fast-uncovered-sampling}, this data structure can be built in time $O(n\log n)$. Moreover, each of the $O(k\log n)$ updates $\Ualive \gets \Ualive \setminus S$ takes time $O(n\log n)$, as $S$ can be computed in $n$ calls to $\contains(S, \cdot)$ and each deletion can be performed in $O(\log n)$ time. Hence, the total time needed to perform all updates to $\Ualive$ is at most $O(nk\log^2 n) = \Ot(nk)$. Therefore, the overall runtime of $\FastSetCover$ is at most $\Ot(mk + nk)$.

    Finally, we verify the desired guarantee on success probability. Each call to $\FindHeavySet$ succeeds with probability $1-1/n^3$, and conditioned on their success, the total number of calls is at most $1+\log_2 k + n = O(n)$. By a union bound, we conclude that $\FastSetCover$ satisfies the desired runtime and approximation guarantees with probability $1-O(1/n^2)$. Moreover, by early termination of $\FastSetCover$, we can ensure that the above bound on the runtime holds deterministically.
\end{proof}

\begin{remark}
    While Lemma~\ref{lem:fast-greedy-cover} yields an $O(\ln n)$-approximation to the minimum set cover, the algorithm can easily be modified to achieve, for any $\eps\in(0,1)$, a $(1+\eps)\ln n$-approximation in time $\widetilde{O}\left(\frac{mk}{\eps^{O(1)}} + nk\right)$.
\end{remark}

\subsection{An $\Ot(n^{\omega} \log \Delta / \eps)$-time Bicriteria Approximation}

The algorithm from Theorem~\ref{thm:fast-log-approx} is efficient when $\OPT_{\textup{size}}$ is small, but its runtime can grow up to $O(n^3)$ in the worst case when $\OPT_{\textup{size}}$ is large. In this subsection, we present a faster algorithm for constructing sparse $\alpha$-navigable graphs that always runs in subcubic time. Specifically, leveraging fast matrix multiplication, we design an $\Ot(n^\omega \log \Delta / \eps)$-time algorithm that returns an $(\alpha, 2\alpha(1+\eps))$-bicriteria $O(\ln n)$-approximation, for any $\eps > 0$.

\paragraph{Algorithm Overview.}Throughout this subsection, we assume $\sfd$ is a metric over $[n]$ with aspect ratio $\Delta$, and without loss of generality, that all distances lie in $[1, \Delta]$. For $\eps>0$, we say that a metric $\sfd$ is \emph{$\eps$-discretized} if all pairwise distances are integer powers of $(1+\eps)$. Our algorithm is motivated by the following observation. Suppose we are given a graph $G = ([n], E)$ and want to verify whether $G$ is $\alpha$-navigable with respect to $\sfd$. A naive approach would check all $\Theta(n^2)$ $\alpha$-navigability constraints separately, which takes $\Theta(n^3)$ time. However, if the metric $\sfd$ is $\eps$-discretized, it turns out that this check can be performed in $O(n^{\omega} \cdot (\log\Delta)/\eps)$ time using fast matrix multiplication.

The key subroutine in our algorithm, $\VerifyNav$ (Algorithm~\ref{alg:verify-nav}), does exactly this. For each of the $ O(\log\Delta/\eps)$ distance scales $(1+\eps)^i$, it creates an $n\times n$ Boolean matrix $B^{(i)}$ indicating pairwise distances that are strictly less than $(1+\eps)^i/\alpha$. It then computes $C^{(i)} = A \cdot B^{(i)}$, where $A$ is the adjacency matrix of the current graph $G$. The matrix $C^{(i)}$ satisfies the following property: for any pair $(s,t)$ such that $\sfd(s,t) = (1+\eps)^i$, the entry $C^{(i)}_{s,t}$ is nonzero if and only if the $\alpha$-navigability constraint from $s$ to $t$ is covered in $G$.

Assume for now that the metric $\sfd$ is $\eps$-discretized; this introduces an additional $(1+\eps)$ factor in the bicriteria guarantee. The main algorithm, $\BuildNav$ (Algorithm~\ref{alg:build-nav}), incrementally constructs an $\alpha$-navigable graph $\bG$ on $\sfd$ by repeatedly invoking $\VerifyNav$ to identify and eliminate uncovered $\alpha$-navigability constraints. For each source node $s \in [n]$, the algorithm maintains: a working set $\calU_s$ of uncovered targets, a sampling budget $\hat{k}_s$, and a counter $c_s$.

In each round, the algorithm \emph{simultaneously} updates the out-neighborhoods of all sources in parallel. For each source node $s$, it samples $\hat{k}_s$ targets uniformly at random from $\calU_s$, adds edges from $s$ to each sampled target, and then calls $\VerifyNav$ to compute all updated uncovered sets $\Unew_s$. If $\Unew_s$ remains nonempty, the counter $c_s$ is incremented. Once $c_s$ reaches $\Omega(\log n)$, the sampling budget $\hat{k}_s$ is doubled and $c_s$ is reset to zero. The key lemma of this section is Lemma~\ref{lemma:coverage-reduction}, which shows that when $\hat{k}_s$ exceeds the minimum degree of $s$ in any $2\alpha$-navigable graph on $\sfd$, the uncovered set $\calU_s$ shrinks by a constant factor with constant probability. Therefore, it takes only $O(\log n)$ rounds for $\calU_s$ to become empty. Since each round takes $O(n^\omega \cdot \log \Delta / \eps)$ time and involves adding at most $\hat{k}_s$ edges per node, we obtain the $O(\ln n)$-approximation guarantee in the desired runtime.

\begin{algorithm}[H]\label{alg:build-nav}
\caption{$\BuildNav(\sfd, n, \alpha,\eps)$}
\KwIn{a metric $\sfd$ on $[n]$, parameters $\alpha\ge1$, $\eps>0$}
\KwOut{an $\alpha$-navigable graph $\bG=([n],\bE)$}
\BlankLine
Round all distances down to $\tilde{\sfd}(s,t) := (1+\eps)^{\floor*{\log_{(1+\eps)}\sfd(s,t)}}$\;
Initialize $\bE = \emptyset$ and set budgets $\hat{k}_s = 1$ and counters $c_s = 0$ for all $s\in[n]$\;
Initialize $\calU_s = \cbr{t \mid t\neq s}$ for all $s$\;
\Repeat{$\bigcup_s \calU_s = \emptyset$}{
    \ForEach{$s\in[n]$ with $\calU_s\neq \emptyset$}{
        Sample $\bu_1,\dots,\bu_{\hat{k}_s}\simiid\Unif(\calU_s)$\;
        Update $\bE\gets \bE\cup \{(s,\bu_i) \mid 1\leq i \leq \hat{k}_s\}$\;
    }
    $\cbr{\Unew_s} \gets \VerifyNav(\tilde{\sfd}, \bG, n, \alpha(1+\eps))$\;
  
    \ForEach{$s\in[n]$ with $\Unew_s\neq \emptyset$}{
    \If{$c_s \geq 120\cdot \log_{16/15} n$}
    {
        Update $\hat{k}_s\gets 2\cdot \hat{k}_s$\;
        Reset $c_s = 0$\;
    }
    \Else{
        Update $c_s\gets c_s + 1$\;
    }
  }
    $\cbr{\calU_s}\gets\cbr{\Unew_s}$\;
}
\Return $\bG$\;
\end{algorithm}

\begin{algorithm}[H]\label{alg:verify-nav}
\caption{$\VerifyNav(\tilde{\sfd}, G, n, \alpha)$}
\KwIn{$\eps$-discretized metric $\tilde{\sfd}$ on $[n]$ with aspect ratio $\Delta$, parameter $\alpha\geq 1$, and graph $G = ([n], E)$}

\KwOut{set of all \(\alpha\)-navigability constraints $(s,t)$ unsatisfied by \(G\)}
\BlankLine
Let $A$ be the adjacency matrix of $G$\;
Initialize $\calU_s=\emptyset$ for all $s\in[n]$\;
\ForEach{$i\in\mathcal [\log_{1+\eps}\Delta]$}{
  Build matrix $B^{(i)}$ with $B^{(i)}_{s,t} = \one\{\tilde{\sfd}(s,t) < (1+\eps)^i / \alpha\}$\;
  Compute $C^{(i)} \leftarrow A\cdot B^{(i)}$\;
}
\ForEach{$s\in[n],\, t\in [n]\setminus \{s\}$} {
    Let $i = \log_{(1+\eps)} \tilde{\sfd}(s,t)$\;
    \If{$C^{(i)}_{s,t} = 0$}{
        Update $\calU_s \gets \calU_s \cup \{(s,t)\}$\;
    }
  }
\Return $\cbr{\calU_s \mid s \in [n]}$\;
\end{algorithm}

We first prove correctness of the verification algorithm $\VerifyNav$.

\begin{claim}\label{claim:fmm-verify}
    If $\tilde{\sfd}$ is an $\eps$-discretized metric on $[n]$ with aspect ratio $\Delta$, then $\VerifyNav(\tilde{\sfd},G, n, \alpha)$ returns the set of $\alpha$-navigability constraints unsatisfied by $G$ in time $O(n^{\omega} \cdot \log\Delta / \eps)$.
\end{claim}
\begin{proof}
    We first check the runtime of $\VerifyNav$. The algorithm has two phases. First, it computes $\log_{1+\eps} \Delta = O(\log\Delta/\eps)$-many multiplications of matrices in $\{0,1\}^{n\times n}$, each of which takes time $O(n^\omega)$. Next, it constructs the lists $\calU_s$ by iterating over all pairs $(s,t)$, which takes time $O(n^2)$. Therefore, the total runtime of $\VerifyNav$ is bounded by $O(n^\omega \cdot (\log\Delta)/\eps)$.

    To check correctness, we fix a pair $(s,t)$ at distance $\sfd(s,t) = (1+\eps)^i$ and verify that $C^{(i)}_{s,t} = 1$ if and only if $G = ([n],E)$ satisfies the $\alpha$-navigability from $s$ to $t$. If indeed $G$ satisfies the constraint, there is some edge $(s,u)\in E$ for which $\sfd(u,t) < \sfd(s,t) / \alpha$, meaning $\sfd(u,t) < (1+\eps)^{i} / \alpha$. Therefore, \[C^{(i)}_{s,t} = \sum_{v=1}^n A_{s,v}\cdot B^{(i)}_{v,t} \geq A_{s,u}\cdot B^{(i)}_{u,t} = 1\] Next, suppose that $G$ does not satisfy the constraint. Then, for any edge $(s,v)\in E$ we must have $\sfd(v,t)\geq(1+\eps)^i / \alpha$, and thus \[C^{(i)}_{s,t} = \sum_{v=1}^n A_{s,v}\cdot B^{i}_{v,t} = 0,\] as desired.
\end{proof}

Next, we show $\eps$-discretizing a metric only affects navigability up to a factor $(1+\eps)$ in the navigability parameter $\alpha$.

\begin{claim}\label{claim:rounding-stability}
Let $\tilde\sfd$ be obtained by rounding down every distance in $\sfd$ to the closest power of $1+\eps$ as in Algorithm~\ref{alg:build-nav}. If $G$ is $\alpha$-navigable on $\tilde\sfd$, then it is $\alpha / (1+\eps)$-navigable on $\sfd$. Moreover, if $H$ is $\alpha$-navigable on $\sfd$, then it is $\alpha/(1+\eps)$-navigable on $\tilde{\sfd}$.
\end{claim}
\begin{proof}
Suppose $G$ is $\alpha$-navigable on $\tilde{\sfd}$. Then, for any $s,t,u\in[n]$ where $\tilde{\sfd}(u,t) < \tilde{\sfd}(s,t) / \alpha$, we have \[\sfd(u,t) < (1+\eps)\cdot \tilde{\sfd}(u,t) < \frac{1+\eps}{\alpha}\cdot \tilde{\sfd}(s,t) \leq \frac{1+\eps}{\alpha}\cdot \sfd(s,t),\] which proves that $G$ is $\alpha/(1+\eps)$-navigable on $\sfd$.

Next, suppose that $H$ is $\alpha$-navigable on $\sfd$. Then, for any $s,t,u\in[n]$ where $\sfd(u,t)<\sfd(s,t)/\alpha$, we have \[\tilde{\sfd}(u,t) \leq \sfd(u,t) < \frac{\sfd(s,t)}{\alpha} < \frac{1+\eps}{\alpha}\cdot \tilde{\sfd}(s,t),\] which proves that $H$ is $\alpha/(1+\eps)$-navigable on $\tilde{\sfd}$.
\end{proof}

Recall from Definition~\ref{def:sc-instance-nav} that $Z_\alpha(s,u) := \{t \mid \sfd(u,t) < \sfd(s,t)/\alpha\}$. The following claim relates $Z_{\alpha}(s,\cdot)$ with $Z_{2\alpha}(s,\cdot)$, and is the key ingredient in our bicriteria approximation analysis.

\begin{claim}\label{claim:reverse-covering}
    For any metric space $([n],\sfd)$ and $\alpha\geq 1$, suppose that $s,x,y,z\in[n]$ are such that $z\in Z_{2\alpha}(s,x)$ and $\sfd(x,y)\leq\sfd(x,z)$. Then, it follows that $z\in Z_\alpha(s,y)$.
\end{claim}
\begin{proof}
    Since $\sfd(x,y)\leq\sfd(x,z)$ we have \[\sfd(y,z)\leq \sfd(x,y) + \sfd(x,z) \leq 2\cdot\sfd(x,z).\] 
    Then, since $z\in Z_{2\alpha}(s,x)$, we know \[\sfd(s,z) > 2\alpha \cdot \sfd(x,z) \geq \alpha \cdot \sfd(y,z)\] which implies $z\in Z_\alpha(s,y)$.
\end{proof}

\begin{lemma}\label{lemma:coverage-reduction}
For any iteration of the main loop in $\BuildNav$, let $s\in[n]$ be a source whose current uncovered set is $\calU_s$ with $N=|\calU_s|$.  Let $k_s$ be the minimum degree of $s$ in any $2\alpha(1+\eps)$‑navigable graph on $\tilde{\sfd}$. If $\hat{k}_s \ge k_s$, then after updating $\Unew_s = \VerifyNav(\tilde{\sfd}, \bG, n, \alpha(1+\eps))$, we have: \[\Pr\left[|\Unew_s| \leq \frac{15}{16}\cdot |\calU_s|\right] \geq 1/15.\]
\end{lemma}
\begin{proof}
$\BuildNav$ samples $\bu_1,\ldots,\bu_{\hat{k}_s}\simiid \calU_s$ and adds all edges $(s,\bu_i)$ to $\bG$, and updates \[\Unew_s = \VerifyNav(\tilde{\sfd}, \bG, n, \alpha(1+\eps)).\] To analyze the size of $\Unew$, observe that since there is a $2\alpha(1+\eps)$-navigable graph on $\tilde{\sfd}$ where $s$ has degree $k_s$, we can find $x_1,\ldots,x_{k_s}\in[n]$ with $\calU_s =\bigcup_{i=1}^{k_s}Z_{2\alpha(1+\eps)}(s,x_i)$. 

Partition $\calU_s = Y_1\sqcup\cdots\sqcup Y_{k_s}$ so each $Y_i\subseteq Z_{2\alpha}(s,x_i)$ is disjoint. For every $i\in[k_s]$, further split $Y_i = C_i\sqcup F_i$ where:
\[
C_i = \text{the } \ceil*{\frac{|Y_i|}{2}}\text{ closest points to }x_i,
\qquad
F_i = Y_i \setminus C_i.
\]
By Claim~\ref{claim:reverse-covering}, whenever $u\in C_i$, the entire set $F_i \cup \{u\}$ is contained in
$Z_{\alpha(1+\eps)}(s,u)$.

Now, define \[\bX_i = \begin{cases}
    1 & \text{if } C_i\cap \{\bu_1,\ldots,\bu_{\hat{k}_s}\} \neq \emptyset \\
    0 & \text{otherwise}
\end{cases}\] and observe that \[|\calU_s| - |\Unew_s| \geq \sum_{i=1}^{k_s} (|F_i|+1)\cdot \bX_i \geq \sum_{i=1}^{k_s} |C_i|\cdot \bX_i := \bA.\] Let $N := |\calU_s|$. We will show that $\Pr[\bA\geq N/16]\geq 1/15$. First, the expectation of each $\bX_i$ is given by: \[\E[\bX_i] = \del{1 - (1-|C_i|/N)^{\hat{k}_s}}\geq \del{1 - \exp(-\hat{k}_s |C_i| / N)}.\] By linearity of expectation, \[\E[\bA] = \sum_{i=1}^{k_s} |C_i|\cdot\E[\bX_i] \geq \sum_{i=1}^{k_s} |C_i| \cdot \del{1 - \exp(-\hat{k}_s |C_i| / N)} \geq \frac{N}{2} - \sum_{i=1}^{k_s} |C_i|\cdot\exp(-k_s |C_i| / N)\] where we have used that $|C_i|\geq |Y_i| / 2$ and $\hat{k}_s \geq k_s$. To analyze the right-most sum, observe that \[|C_i|\cdot \exp(-k_s|C_i|/N) \leq \frac{N}{k_s}\cdot \max_{t\geq 0}\del{te^{-t}} = \frac{N}{ek_s}.\]

This immediately gives \[\E[\bA] \geq \frac{N}{2} - \frac{N}{e} > \frac{N}{8}.\]

By Markov's inequality, we conclude that $\Pr[\bA \geq N/16] \geq 1/15$, as desired.
\end{proof}

\begin{lemma}\label{lemma:budget-upper-bound}
Consider any source $s\in[n]$ in the execution of $\BuildNav$. Let $k_s$ be the minimum degree of $s$ in any $2\alpha(1+\eps)$-navigable graph on $\tilde{\sfd}$. Then, with probability at least $1 - 1/n^2$, the set $\calU_s$ becomes empty while $\hat{k}_s \leq 2k_s$.
\end{lemma}

\begin{proof}
Fix a source $s \in [n]$, and suppose that $\calU_s$ is still nonempty at the first moment that $\hat{k}_s \geq k_s$. At this point, $\hat{k}_s \leq 2k_s$, and we will show that with high probability, $\calU_s$ will become empty without updating $\hat{k}_s$ further. 

For $R = 120\cdot \log_{16/15} n$, the execution of $\BuildNav$ involves at most $R$ rounds of adding $\hat{k}_s$ random edges from $s$. For each round $i \in [R]$, define an indicator variable $\bY_i$ that is equal to $1$ if $|\Unew_s| \leq (15/16)\cdot|\calU_s|$ at the end of the $i$-th round, and is equal to $0$ otherwise. By Lemma~\ref{lemma:coverage-reduction}, we have $\Pr[\bY_i = 1] \geq 1/15$, independently for each round.

We say that round $i$ is \emph{good} if $\bY_i = 1$. Let $\bY := \sum_{i=1}^{R} \bY_i$ be the number of good rounds. Then $\E[\bY] \geq R/15 = 8\log_{16/15} n$. By a standard Chernoff bound, we have
\[
\Pr\left[\bY \leq\frac{\E[\bY]}{2}\right] \leq \exp\left(-\frac{\E[\bY]}{8}\right) \leq \exp(-2\log n) \leq \frac{1}{n^2}.
\]
Thus, with probability at least $1 - 1/n^2$, there are strictly more than $\log_{16/15} n$ good rounds. Since $|\calU_s| \leq n$ at the beginning, and each good round reduces the size of $\calU_s$ by a factor of $15/16$, then after more than $\log_{16/15} n$ such rounds, we will have $|\calU_s| = 0$. Therefore, with probability at least $1 - 1/n^2$, we will have $\calU_s = \emptyset$ when $\hat{k}_s \leq 2k_s$.
\end{proof}

We are now ready to prove Theorem~\ref{thm:bicriteria-alg}.

\begin{proof}[Proof of Theorem~\ref{thm:bicriteria-alg}]
Let $(P = [n], \sfd)$ be an $n$-point metric with aspect ratio $\Delta$, and fix $\alpha \geq 1$ and $\eps > 0$. Let $\tilde{\sfd}$ be the $\eps$-discretized version of $\sfd$ obtained by rounding each distance down to the nearest power of $1+\eps$, as done in $\BuildNav$. We will show that with probability $1 - o(1)$, $\BuildNav$ outputs an $(\alpha(1+\eps), 2\alpha(1+\eps))$-bicriteria $O(\ln n)$-approximation to the sparsest navigable graph on $\tilde{\sfd}$, in time $\Ot(n^\omega \cdot \log \Delta / \eps)$. By Claim~\ref{claim:rounding-stability}, this output is also an $(\alpha, 2\alpha(1+\eps)^2)$-bicriteria $O(\ln n)$-approximation to the sparsest navigable graph on $\sfd$, and a rescaling of $\eps$ gives the desired guarantee.

Correctness follows from the construction: by Claim~\ref{claim:fmm-verify}, the graph $\bG$ returned upon termination of $\BuildNav$ is guaranteed to be $\alpha(1+\eps)$-navigable on $\tilde{\sfd}$. We will now analyze the sparsity of the graph $\bG$.

For each $s \in [n]$, let $k_s$ denote the minimum degree of $s$ in any $2\alpha(1+\eps)$-navigable graph on $\tilde{\sfd}$. By Lemma~\ref{lemma:budget-upper-bound} and a union bound over all $s$, we have that with probability at least $1 - 1/n$, the budget $\hat{k}_s$ remains at most $2k_s$ throughout the execution. In the remainder of the proof, we will condition on this event.

As the budget $\hat{k}_s$ is held fixed for $O(\log n)$ rounds, and each round adds at most $\hat{k}_s$ edges from $s$, the number of edges added from $s$ at the budget level $\hat{k}$ is $O(\hat{k}_s \cdot \log n)$. Since $\hat{k}_s$ is scaled geometrically up to at most $2k_s$, we conclude that $\deg_{\bG} s = O(k_s \log n)$. Summing over all sources $s \in [n]$, the total number of edges in the final graph is at most
\[
O(\log n) \cdot \sum_{s=1}^n k_s = O(\log n) \cdot \OPT_{2\alpha(1+\eps)},
\]
where $\OPT_{2\alpha(1+\eps)}$ denotes the edge count of the sparsest $2\alpha(1+\eps)$-navigable graph on $\tilde{\sfd}$.

Similarly, the maximum degree is at most $O(\log n) \cdot \max_s k_s$, which gives the desired bicriteria guarantee under both the $\maxdeg$ and $\avgdeg$ objectives.

It remains to bound the runtime, which is given by the number of calls to $\VerifyNav$. Since each source $s\in[n]$ takes $O(\log k_s \cdot \log n) = O(\log^2 n)$ rounds of adding edges before $\calU_s$ becomes empty, the total number of calls to $\VerifyNav$ is at most $O(\log^2 n)$. By the runtime guarantee of Claim~\ref{claim:fmm-verify}, we conclude that $\BuildNav$ runs in time
\[
O(n^\omega \cdot \log^2 n \cdot \log\Delta / \eps) = \Ot(n^\omega \cdot \log \Delta / \eps).
\]

The above bound on the runtime occurs with probability $1-1/n$, but it can be enforced deterministically by early termination of $\BuildNav$.
\end{proof}
\newcommand{\bsfd}{\boldsymbol{\sfd}}
\newcommand{\dpath}{\sfd_{\textup{path}}}

\section{An $\Omega(n^2)$ Query Lower Bound}\label{sec:lower-bounds}

Theorem~\ref{thm:fast-log-approx} presents an algorithm for producing an $O(\ln n)$-approximation to the sparsest $\alpha$-navigable graph, running in time $O(n\cdot \OPT_{\textup{size}})$. Since every $\alpha$-navigable graph must be connected and thus contain $\Omega(n)$ edges, this bound on the runtime can be no better than $O(n^2)$. In this section, we show that such a quadratic overhead is unavoidable: any algorithm achieving even an $o(n)$-approximation to the sparsest $\alpha$-navigable graph must make $\Omega(n^2)$ distance queries.

\begin{theorem}\label{thm:query-lb}
    Let $\calA$ be any randomized algorithm with query access to a metric $\sfd$ on $n$ points that (with constant probability) outputs an $o(n)$-approximation to the sparsest navigable graph on $\sfd$, under either the $\maxdeg$ or $\avgdeg$ objective. Then, $\calA$ must make $\Omega(n^2)$ queries to $\sfd(\cdot,\cdot)$.
\end{theorem}

To prove Theorem~\ref{thm:query-lb}, we construct a distribution over a family of path-like metrics, with a random hidden ``shortcut'' that must appear as an edge in order to satisfy navigability. By Yao's minimax principle, it will suffice to show the stated lower bound for deterministic algorithms $\calA$.

\begin{definition}[Perturbed Path Metric]\label{def:perturbed-path}
    We specify a distribution $\calD$ over metrics $\bsfd$ on $[n]$ as follows.
    \begin{itemize}
        \item For all distinct $i,j\in[n]$, set \[\bsfd(i,j) = \dpath(i,j) := 1+\frac{|i-j|}{n-1}\]
        
        \item Sample $\{\bi^\ast, \bj^\ast\}\sim\binom{[n]}{2}$ uniformly at random, and overwrite $\bsfd(\bi^\ast, \bj^\ast) = \bsfd(\bj^\ast,\bi^\ast) = 1$
    \end{itemize}
    
    Since all nonzero distances are in $[1,2]$, the distance function $\bsfd$ automatically satisfies the triangle inequality and is a metric.
        
\end{definition}

\begin{lemma}\label{lem:sparse-nav-path}
    There is a navigable graph on $\bsfd$ with maximum degree $\le 3$.
\end{lemma}
\begin{proof}
    Let $G = P_n\cup \{\bi^\ast, \bj^\ast\}$, where $P_n$ denotes the (undirected) path graph on $[n]$. Observe that $G$ is navigable on $\dpath$, since for any distinct $s,t$, we have that $|u-t| < |s-t|$ for some $u = s\pm 1$.

    Now consider $\bsfd$, the perturbed metric where a single pair $\{\bi^\ast, \bj^\ast\}$ has its distance reduced to $1$. The only navigability constraints $(s,t)$ that may differ between $\bsfd$ and $\dpath$ are those where $t \in \{\bi^\ast, \bj^\ast\}$. We perform a brief casework on these constraints. \begin{itemize}
        \item If $t = \bj^\ast$ and $s = \bi^\ast$, then there is a direct edge $(\bi^\ast, \bj^\ast)$ in $G$.
        \item Suppose $t = \bj^\ast$ and $s\neq \bi^\ast$. Then, by navigability of $G$ on $\dpath$, there is an edge $(s,u)$ in $G$ such that \[\bsfd(u,t) \le \dpath(u,t) < \dpath(s,t) = \bsfd(s,t).\] 
        \item The case that $t = \bi^\ast$ can be handled identically to the above cases.
    \end{itemize}

    Hence, $G$ is navigable on $\bsfd$, and it has maximum degree $\le 3$.
\end{proof}

\begin{lemma}\label{lem:quadratic-lb}
    Let $\calA$ be any deterministic algorithm that makes $o(n^2)$ queries to $\bsfd$ and outputs a graph $G$ with $o(n^2)$ edges. With probability $1-o(1)$ over $\bsfd \sim \calD$, the graph $\bG$ is not navigable on $\bsfd$.
\end{lemma}

\begin{proof}
Let $\bS \subset \binom{[n]}{2}$ denote the set of distance queries made by $\calA$. Define the event
\[
\calbE := \left\{ \{\bi^\ast, \bj^\ast\} \notin \bS \right\}.
\]
Since $\{\bi^\ast, \bj^\ast\}$ is drawn uniformly at random from $\binom{[n]}{2}$ and $|\bS| = o(n^2)$, we have $\Pr[\calbE] = 1 - o(1)$.

Let $\bG = ([n], \bE)$ be the output graph of $\calA$, which is entirely determined by the responses to the queries in $\bS$. Conditioning on $\calbE$, the pair $\{\bi^\ast,\bj^\ast\}$ is uniform over $\binom{[n]}{2}\setminus \bS$. Since $|\bE| = o(n^2)$, the algorithm $\calA$ includes the pair $(\bi^\ast, \bj^\ast)$ in $\bE$ only with probability $o(1)$. Therefore,
\[
\Prx_{\bsfd} \left[ (\bi^\ast, \bj^\ast) \in \bE \right] \leq \Prx_{\sfd}\left[\bar{\calbE}\right] + \Prx_{\bsfd} \left[ (\bi^\ast, \bj^\ast) \in \bE \mid \calbE \right] = o(1).\]

Since $\bsfd(\bi^\ast, \bj^\ast) = 1$ is the minimum distance in the metric, any navigable graph on $\bsfd$ must contain the edge $(\bi^\ast, \bj^\ast)$. Therefore,
\[
\Prx_{\bsfd} \left[ \bG \text{ is navigable on } \bsfd \right] = o(1),
\]
as claimed.
\end{proof}

\begin{proof}[Proof of Theorem~\ref{thm:query-lb}]
By Yao’s minimax principle, it suffices to prove the lower bound against deterministic algorithms under the random input metric $\bsfd\sim\calD$ from Definition~\ref{def:perturbed-path}. Let $\calA$ be any deterministic algorithm that makes $o(n^2)$ queries to the metric $\bsfd$ and outputs a graph $\bG = ([n], \bE)$ that, with constant probability, is an $o(n)$-approximation to the sparsest navigable graph on $\bsfd$, under either the $\maxdeg$ or $\avgdeg$ objective.

By Lemma~\ref{lem:sparse-nav-path}, there exists a navigable graph on $\bsfd$ with maximum degree at most $3$, and hence $O(n)$ total edges. Therefore, any $o(n)$-approximation to the sparsest navigable graph on $\bsfd$ (under either objective) must contain $o(n^2)$ edges.

However, Lemma~\ref{lem:quadratic-lb} shows that no deterministic algorithm making $o(n^2)$ queries can, with constant probability, output a navigable graph on $\bsfd$ with $o(n^2)$ edges. Therefore, any such algorithm $\calA$ must make $\Omega(n^2)$ queries.
\end{proof}

\bibliographystyle{alpha} 
\bibliography{references, waingarten} 

\appendix 

\section{Missing Proofs from Section~\ref{sec:upper-bounds}} \label{appendix:missing-proofs-ub}

\harmonicbalancing*
\begin{proof}
    Let $\ell \in [L]$ maximize the quantity $\ell\cdot\alpha_{\ell}$. For each $j\in[L]$, we have
    \[
        \alpha_j \leq \frac{\ell\cdot\alpha_{\ell}}{j}.
    \]
    Summing this inequality over all $j \in [L]$ yields
    \[
        A = \sum_{j=1}^{L}\alpha_j \leq \sum_{j=1}^{L}\frac{\ell\cdot\alpha_{\ell}}{j} 
        = (\ell\cdot\alpha_{\ell})\cdot H_L,
    \]
    as desired.
\end{proof}

\sizetesting*
\begin{proof}
    Suppose that $|S|\geq \alpha N$. Then, $\Ex[|S\cap \bX|] \geq \alpha T$, and by a lower-tail multiplicative Chernoff bound, \[\Prx[|S\cap \bX| < \alpha T / 2] \leq\exp\del{-\frac{\alpha T}{8}} \leq \exp\del{-\frac{\alpha T}{12}}.\] Suppose that $|S|\leq \alpha N / 4.$ Since $\Pr[|S\cap \bX| \geq \alpha T/2]$ increases monotonically with $|S|$, it suffices to consider the setting of $|S| = \alpha N / 4$, for which $\Ex[|S\cap \bX|] = \alpha T / 4$. By an upper-tail multiplicative Chernoff bound, \[\Pr[|S\cap \bX| \geq \alpha T / 2] \leq \exp\del{-\frac{\alpha T}{12}},\] completing the proof.
\end{proof}

\fastuncoveredsampling*
\begin{proof}
We implement $\Ualive$ using a complete binary search tree on $2^{\ceil*{\log_2 n}}$ leaves, where each element of $\calU$ corresponding to one leaf. To initialize $\Ualive = \calU$, we store at each internal node $v$ a counter
\[
  c_v = \abs{\calU\cap \text{(leaves in the subtree of }v)},
\]
which takes time $O(n)$.

\begin{itemize}
    \item \textbf{Checking emptiness.} $\Ualive=\emptyset$ if and only if the root $r$ has $c_r=0$, which can be tested in $O(1)$ time.

    \item \textbf{Deletion.} To delete an element $x$ from $\Ualive$, we first confirm that $x\in\Ualive$ by searching for its leaf in the tree and checking that $c_x = 1$. If so, then for each vertex $v$ in the unique path from the root to that leaf, we decrement $c_v$ by $1$.  Since the tree has height $O(\log n)$, this takes $O(\log n)$ time.

    \item \textbf{Random sampling.} To sample a uniformly random element of $\Ualive$, we traverse a random root-to-leaf path as follows. Suppose we are at node $v$ with children $v_{\textup{left}},v_{\textup{right}}$.  We choose to go to
\[
    \begin{cases}
        v_{\textup{left}} & \text{with probability } c_{v_{\textup{left}}}/c_v \\
        v_{\textup{right}} & \text{with probability } c_{v_{\textup{right}}} / c_v
    \end{cases}
\]
Since $c_v=c_{v_{\textup{left}}}+c_{v_{\textup{right}}}$, this selects a leaf in proportion to the number of elements in $\Ualive$ beneath each child.  Upon reaching a leaf, we return its corresponding element, which is a uniformly random sample from $\Ualive$.  Each step takes $O(1)$ time, and there are $O(\log n)$ steps, so sampling runs in $O(\log n)$ time.
\end{itemize}
This completes the proof.
\end{proof}

\end{document}